\documentclass[a4paper,UKenglish,cleveref, autoref, thm-restate]{lipics-v2021}

\hideLIPIcs

\title{Rule Rewriting Revisited: A Fresh Look at Static Filtering for Datalog and ASP}

\author{Philipp {Hanisch}}{Knowledge-Based Systems Group, TU Dresden}{philipp.hanisch1@tu-dresden.de}{https://orcid.org/0000-0003-3115-7492}{}
\author{Markus {Krötzsch}}{Knowledge-Based Systems Group, TU Dresden}{markus.kroetzsch@tu-dresden.de}{https://orcid.org/0000-0002-9172-2601}{}

\authorrunning{P. Hanisch and M. Krötzsch}

\Copyright{Philipp Hanisch and Markus Krötzsch}



\category{} 

\relatedversiondetails[linktext=10.4230/LIPIcs.ICDT.2026.5, cite=DBLP:conf/icdt/HanischK25]{This paper is the technical report for the paper with the same title
	appearing in the Proceedings of the 29th International Conference on Database Theory (ICDT 2026)}{https://doi.org/10.4230/LIPIcs.ICDT.2026.5}

\acknowledgements{
	This work is partly supported by 
		Deutsche Forschungsgemeinschaft (DFG, German Research Foundation)
			in project number 389792660 (TRR 248, \href{https://www.perspicuous-computing.science/}{Center for Perspicuous Systems})
			and 390696704 (CeTI Cluster of Excellence) and
		by the Bundesministerium für Forschung, Technologie und Raumfahrt (BMFTR, Federal Ministry of Research, Technology and Space)
			in the \href{https://www.scads.de/}{Center for Scalable Data Analytics and Artificial Intelligence} (ScaDS.AI)
			and in DAAD project 57616814 (\href{https://secai.org/}{SECAI, School of Embedded Composite AI}) as part of the program Konrad Zuse Schools of Excellence in Artificial Intelligence.}

\nolinenumbers 

\EventEditors{Balder ten Cate and Maurice Funk}
\EventNoEds{2}
\EventLongTitle{29th International Conference on Database Theory (ICDT 2026)}
\EventShortTitle{ICDT 2026}
\EventAcronym{ICDT}
\EventYear{2026}
\EventDate{March 24--27, 2026}
\EventLocation{Tampere, Finland}
\EventLogo{}
\SeriesVolume{365}
\ArticleNo{5}

\usepackage{subcaption}
\usepackage{mathtools}
\usepackage[linesnumbered,ruled,vlined,noresetcount]{algorithm2e}
\usepackage{multirow}

\usepackage{xspace}
\usepackage{thmtools}

\newcommand{\ghost}[1]{\raisebox{0pt}[0pt][0pt]{\makebox[0pt][l]{#1}}}

\newcommand{\alglineref}[1]{L\ref{#1}} 

\usepackage{xcolor}
\definecolor{darkgreen}{HTML}{36AB14}

\usepackage{todonotes}

\newcommand{\opfont}[1]{\text{\sf{#1}}} \usepackage{bm}
\renewcommand{\vec}[1]{\bm{#1}}
\newcommand{\ltuple}{\langle}
\newcommand{\rtuple}{\rangle}
\newcommand{\tuple}[1]{\ltuple{#1}\rtuple}
\newcommand{\narrow}[1]{\,{#1}\,} \newcommand{\blank}{\text{\textvisiblespace}} 

\newcommand{\natnums}{\mathbb{N}}

\newcommand{\card}[1]{\vert #1 \vert}

 \newcommand{\naf}{\mathit{not}\,} 
\newcommand{\shortvee}{\,{\vee}\,}

\newcommand{\stablemods}[2]{\opfont{sm}(#1,#2)} 

\makeatletter
\providecommand*{\cupdot}{\mathbin{\mathpalette\@cupdot{}}}
\newcommand*{\@cupdot}[2]{\ooalign{$\m@th#1\cup$\cr
		\hidewidth$\m@th#1\cdot$\hidewidth
	}}
\makeatother

\newcommand{\sigPred}{\mathbf{P}}
\newcommand{\sigPredFilter}{\mathbf{F}}
\newcommand{\sigPredOutput}{\mathbf{P}_{\opfont{out}}}

\newcommand{\sigPredIDB}{\mathbf{P}_{\opfont{IDB}}}
\newcommand{\sigPredStrat}{\mathbf{P}_{\opfont{str}}}
\newcommand{\sigCons}{\mathbf{C}}
\newcommand{\sigVar}{\mathbf{V}}

\newcommand{\genFilters}{\mathbf{G}} 

\newcommand{\arity}[1]{\opfont{ar}(#1)} \newcommand{\vars}[1]{\opfont{var}(#1)} \newcommand{\sigNumPos}{\mathbf{N}}
\newcommand{\numpos}[1]{{\setlength{\fboxsep}{1pt}\framebox{\scalebox{0.9}{\makebox[0pt][l]{\phantom{1}}$\boldsymbol{#1}$}}}} \newcommand{\tonumpos}[1]{{\setlength{\fboxsep}{1pt}\framebox{\scalebox{0.9}{\makebox[0pt][l]{\phantom{1}}$#1$}}}} 

\newcommand{\ground}[1]{\opfont{gr}(#1)} 

\newcommand{\inter}[1]{\mathcal{#1}}  \newcommand{\adatabase}{\inter{D}}  

\newcommand{\reprFunc}{\opfont{rep}} \newcommand{\repr}[1]{\reprFunc(#1)}

\newcommand{\filterFormulas}[1]{\mathcal{F}_{#1}}

\newcommand{\fnFilterExpr}[1]{\opfont{flt}(#1)}

\newcommand{\approxmodels}{\mathrel|\joinrel\approx}
\newcommand{\approxequiv}{\approxeq}

\newcommand{\arule}{\rho} \newcommand{\atom}[1]{\underline{#1}}  

  \newcommand{\aprogram}{P}

\newcommand{\complclass}[1]{\text{{\sc #1}}\xspace} 

\newcommand{\PTime}{\complclass{P}}

\newcommand{\coNP}{co\complclass{NP}}

\newcommand{\shBodyNormal}{B_{\bar{\sigPredFilter}}}
\newcommand{\shBodyNaf}{B_{\bar{\sigPredFilter}}^-}
\newcommand{\shBodyFilterG}{G_{\sigPredFilter}}
\newcommand{\shBodyFilterGRewrite}{\psi}
\newcommand{\shBody}{\shBodyNormal \land \shBodyNaf \land \shBodyFilterG}

\newcommand{\shHead}{h(\vec{x})}
\newcommand{\shBodyNormali}[1]{B_{#1,\bar{\sigPredFilter}}}
\newcommand{\shBodyNafi}[1]{B_{#1,\bar{\sigPredFilter}}^-}
\newcommand{\shBodyFilterGi}[1]{G_{#1,\sigPredFilter}}

\newcommand{\shHeadi}[1]{h_#1(\vec{x}_#1)}

\newcommand{\satisfy}[2]{#1 \in #2^\adatabase}

\newcommand{\predOddK}[1]{\mathbf{1}_{#1}}
\newcommand{\predEvenK}[1]{\mathbf{0}_{#1}}

\newcommand{\const}[1]{\mathtt{#1}\xspace}
\newcommand{\predEquals}{\mathrel{\overset{\cdot}{=}}}

\begin{document}

	\maketitle

	\begin{abstract}
\emph{Static filtering} is a data-independent optimisation method 
for Datalog, which generalises algebraic query rewriting techniques from relational databases.
In spite of its early discovery by Kifer and Lozinskii in 1986, the method has been overlooked
in recent research and system development, and special cases are being rediscovered independently.
We therefore recall the original approach, using updated terminology and more general filter predicates
that capture features of modern systems, and we show how to extend its applicability to answer set programming (ASP).
The outcome is strictly more general but also more complex than the classical approach:
double exponential in general and single exponential even for predicates of bounded arity.
As a solution, we propose tractable approximations of the algorithm that can still yield much improved 
logic programs in typical cases, e.g., it can improve the performance of rule systems over real-world data in the order of magnitude.
\end{abstract}
 	
	\section{Introduction}

Besides its many other advantages, the declarative nature of logic-based rule languages also
enables effective optimisation through logically equivalent rewritings.
Of course, already for plain Datalog, logical equivalence is undecidable, and highly complex
in decidable special cases \cite{BKR15:reasonQuery}.
But many concrete transformations that guarantee logical equivalence have been proposed, ranging
from the popular \emph{magic sets} \cite{DBLP:conf/pods/BancilhonMSU86,DBLP:conf/sigmod/MumickFPR90,DBLP:conf/sigmod/MumickP94}
to recent proposals \cite{DBLP:conf/ppdp/TekleL10,DBLP:journals/corr/abs-1909-08246,DBLP:conf/sigmod/WangK0PS22,DBLP:journals/tplp/ZanioloYDSCI17}.
Like these examples, many rule rewritings are \emph{static} optimisations, which do not depend
on the concrete set of facts (data) that is to be processed.

One of the classical proposals of this kind is \emph{static filtering}, the generalisation of
\emph{selection pushing} methods from relational databases to Datalog,
introduced by Kifer and Lozinskii \cite{KiferLozinskii:StaticFiltering86,KiferLozinskii:StaticFilteringImpl87,KiferLozinskii:StaticFiltering90}.
The underlying principle of enforcing restrictions (``filters'') on intermediate results as early
as possible in the computation is a tried and tested paradigm in databases, and does -- in contrast
to, e.g., magic sets -- not change the general structure of derivations fundamentally.
Moreover, static filtering strongly increases the effectivity of the related method of
projection pushing \cite{KiferLozinskii:StaticFiltering90}.
Static filtering can enable polynomial (data complexity) or exponential (combined complexity) performance improvements (see Examples~\ref{ex_motivation} and \ref{ex_bounded_reach_sfrewrite_projected}).

Surprisingly, the journal paper of Kifer and Lozinskii has attracted less than 50 citations
in the past 35 years.\footnote{Google Scholar, \url{https://scholar.google.com/scholar?cites=13499523163799224695}, retrieved 8 September 2025}
Actual uses of the method are rarely reported \cite{DBLP:conf/rr/Billig08,DBLP:journals/infsof/FernandesBPW97,DBLP:conf/caise/LiuLYWH09}, while most mentions consider it distantly related work.
Is this basic optimisation approach maybe so fundamental that it is not even mentioned by implementers?
Or is it known and used under another name?

\begin{example}\label{ex_motivation}
We ran a small experiment to find out, using the following Datalog rules (see Section~\ref{sec_prelims} for a
formal introduction to Datalog, and the literature for more details \cite{Kroetzsch2025:Datalog,Maier+:DatalogHistory18}):
\begin{align}
p(0,\ldots,0,0,\const{a}) . &\qquad p(1,\ldots,1,0,\const{b}) . \label{eq_testeval_in}\\
p(x_1,\ldots, x_i,1,0,\ldots,0,y) &\textstyle \leftarrow p(x_1,\ldots, x_i,0,1,\ldots,1,y) \hspace{.8cm} \text{for all $i\in\{1,\ldots,\ell\}$} \label{eq_testeval_step}\\
\textit{out}(y) & \leftarrow p(x_1,\ldots,x_\ell,y)\wedge y\predEquals \const{b} \label{eq_testeval_out}
\end{align}
where $p$ is an $(\ell+1)$-ary predicate with $\ell\geq 1$, $\predEquals$ denotes equality, $y$ and $x_k$ are variables, and $0$, $1$, $\const{a}$, and $\const{b}$ are constants.
Rules \eqref{eq_testeval_step} implement a binary counter over $\ell$ bits, so exponentially many $p$-facts are inferred from the facts \eqref{eq_testeval_in}.
Optimisation is possible if we are only interested in inferences for predicate $\textit{out}$:
then the precondition $y\predEquals \const{b}$ can be added to the rules \eqref{eq_testeval_step}, so that just one new $p$-fact follows.
Static filtering produces this rewriting.
\end{example}

\begin{table}[t]
\caption{Runtimes for example program \eqref{eq_testeval_in}--\eqref{eq_testeval_out} with $\ell=19$ (median of five runs, timeout at 5min, evaluation system: Linux, AMD Ryzen 7 PRO 5850U, 16 GiB RAM)}\label{tab_testeval}
\mbox{}\hfill \begin{tabular}{rrrrr}
    & \textbf{Souffl\'{e}} {\footnotesize v2.5} & \textbf{Nemo} {\footnotesize v0.8.1} & \textbf{Clingo} {\footnotesize v5.8.0} & \textbf{DLV} {\footnotesize v2.1.2} \\\hline \\[-2ex]
Original  & 1214ms & $>$5min & 1104ms & 74579ms\\
Rewritten & 24ms& 99ms & 8ms & 3ms\\
\end{tabular}\hfill\mbox{}
\end{table}
Many modern rule systems let users specify output predicates. For our experiment, we considered Datalog engines
Souffl\'{e} \cite{Jordan+:Souffle16} (syntax {\tt .output out}) and Nemo \cite{Ivliev+:Nemo2024} (syntax {\tt @export out :- csv\{\}.}),
and ASP engines Clingo \cite{Gebser+:clingo2019} and DLV \cite{Alviano+:DLV2:17} (syntax {\tt \#show out/1.} for both).
Table~\ref{tab_testeval} shows runtimes for the original program and the optimised version.
Evidently, each tool benefits from the optimisation, yet none implements it by default.

Why is such a natural optimisation, considered standard in relational query optimisation, ignored in
modern rule systems? Research culture may be a reason. Typical benchmarks for comparing systems 
are already optimised, so static optimisations offer no benefits. They are likely more effective with less polished user
inputs, especially during development and experimentation.
Moreover, static filtering has not attained the popularity of other approaches, especially magic sets and
semi-naive evaluation, and may not be known to many implementers. The original description relies on \emph{system graphs} as an auxiliary concept that many readers may not know today,
yet we are not aware of modern accounts or textbook explanations of the method.

Another reason might be practicality. Kifer and Lozinskii found the method to be exponential in the worst case -- as hard as the Datalog reasoning task it aims to optimise -- and did not suggest tractable variants.
It is also left open how static filtering generalises to further filter expressions (Kifer and Lozinskii only consider binary relations like $=$, $\neq$, and $\leq$), and to non-monotonic negation or ASP.

 In this work, we therefore revisit static filtering and introduce a generalised rewriting that supports arbitrary filters 
 (Section~\ref{sec_sf}). Analysing the complexity of our method (Section~\ref{sec_sf_steps}), we find further exponential increases over
 the prior special case, which motivates our design of a tractable variant that is still reasonably general (Section~\ref{sec_fentailment}).
 Finally, we show how static filtering can be extended to rules with negation and to ASP (Section~\ref{sec_negation}), before
 comparing closely related works (Section~\ref{sec_rel_work}).
Detailed proofs for all claims are found online \cite{hanisch2026rulerewritingrevisitedfresh}.

 \section{Preliminaries}\label{sec_prelims}

We consider a signature based on mutually disjoint, countably infinite sets of \emph{constants} $\sigCons$, \emph{variables} $\sigVar$,
and \emph{predicates} $\sigPred$, where 
each predicate $p\in\sigPred$ has \emph{arity} $\arity{p}\geq 0$.
\emph{Terms} are elements of $\sigVar \cup \sigCons$.
We write bold symbols $\vec{t}$ for lists $t_1,\ldots,t_{|\vec{t}|}$ of terms (and for terms of special types, such as lists of variables $\vec{x}$).
Lists are treated like sets when order is not relevant, so we may write, e.g., $\vec{x}\subseteq\sigVar$.
By $\vars{E}$ we generally denote the set of variables in an expression $E$.
For a mapping $\sigma \colon \vec{x} \to \sigCons$ and expression $E$, we obtain $E\sigma$ by simultaneously replacing all occurrences of $x \in \vec{x}$ by $\sigma(x)$.

\paragraph*{Rules and programs}
An \emph{atom} $\atom{a}$ is an expression $p(\vec{t})$ with $p\in\sigPred$, $\vec{t}$ a list of terms, and $\arity{p} = \card{\vec{t}}$.
A (Datalog) \emph{rule} $\arule$ is a formula $H\leftarrow B$,
where the \emph{head} $H$ is an atom, the \emph{body} $B$ is a conjunction of atoms, and all variables are implicitly universally quantified.
We require that all variables in $H$ also occur in $B$ (\emph{safety}).
Conjunctions of atoms may be treated as sets of atoms.
A (Datalog) \emph{program} $\aprogram$ is a finite set of rules.
A predicate $p$ is an \emph{EDB atom} in $\aprogram$ if it only occurs in rule bodies,
and an \emph{IDB atom} if it occurs in some rule head.\footnote{These terms originate from \emph{extensional}/\emph{intentional database}.}

\paragraph*{Semantics}
A \emph{fact} for predicate $p$ is an atom $p(\vec{c})$ with $\vec{c}\subseteq\sigCons$.
A \emph{database} $\adatabase$ for a program $\aprogram$ is a potentially infinite set of
facts for EDB predicates of $\aprogram$.
We allow $\adatabase$ to be infinite, so as to
accommodate conceptually infinite built-in relations, such as $\leq$. Practical systems typically evaluate such built-ins on demand and syntactically
ensure that infinite built-ins do not lead to infinite derivations, e.g., by requiring that variables in built-ins also occur
in body atoms with finite predicates. Such concerns are unimportant to our results, so we can unify
(given) input facts and (computed) built-in relations.

The (unique) \emph{model} $\inter{M}$ for program $\aprogram$ and database $\adatabase$ is the least set of facts such that
(1) $\adatabase\subseteq\inter{M}$, and
(2) for every rule $\arule\in\aprogram$ with variables $\vec{x}$, and every mapping $\sigma:\vec{x}\to\sigCons$, if $\sigma(B)\subseteq\inter{M}$ then $\sigma(H)\subseteq\inter{M}$.
If $\sigma(B)\subseteq\inter{M}$, we also call $\sigma$ a \emph{match} of $\arule$ on $\inter{M}$.
Models can be equivalently defined through iterated rule applications, proof trees, or as least models of $\aprogram$ viewed as a first-order logic theory \cite{Alice}.

For a rule $H \leftarrow B$ with variables $\vec{x}$, a mapping $\sigma:\vec{x}\to\sigCons$ is \emph{applicable} to a database $\adatabase$ if $B\sigma \subseteq \adatabase$ and $H\sigma \nsubseteq \adatabase$.
A set of facts $\adatabase$ is \emph{closed under a program} $P$, written $\adatabase \models P$, if there is no $\arule \in P$ with an applicable mapping $\sigma$.
For a predicate $p$, let $p^\adatabase = \{ \vec{c} \mid p(\vec{c}) \in \adatabase \}$ denote the tuples of $p$-facts in $\adatabase$.

\paragraph*{Normal form}
To simplify presentation, we require that rules contain only variables, no constants, and that
atoms in rules do not contain repeated variables. This normal form can be established by defining auxiliary EDB predicates
$(\blank\narrow\predEquals \blank)^\adatabase=\{\tuple{\const{c},\const{c}}\mid \const{c}\in\sigCons\}$ and $(\blank\narrow\predEquals \const{d})^{\adatabase}=\{\const{d}\}$ 
for every constant
$\const{d}\in\sigCons$.
Now every occurrence of $\const{d}\in\sigCons$ in a rule is replaced by a fresh variable $x$, and the atom $x\narrow\predEquals \const{d}$ is 
added to the rule body. Similarly, every non-first occurrence of a variable $x$ in a body atom is replaced
by a fresh variable $x'$, and the atom $x\narrow\predEquals x'$ is added to the rule body.
Similar normalisations can be used for any built-in function that reasoners might support,
e.g., arithmetic functions as in a rule $p(x+y)\leftarrow q(x,y)$ can be rewritten as
$p(z)\leftarrow q(x,y)\wedge z\narrow\predEquals x{+}y$ with $(\blank\narrow\predEquals \blank{+}\blank)^\adatabase=\{\tuple{\const{l},\const{m},\const{n}} \in \sigCons^3\mid \const{l}=\const{m}+\const{n}\}$.
As before, the infinite EDB predicate is just a conceptual model for a real systems' on-demand computation
of its supported built-in functions.

\paragraph*{Outputs and filters}
Given a program $\aprogram$, we consider distinguished sets $\sigPredOutput\subseteq\sigPred$ of \emph{output predicates} and
$\sigPredFilter\subseteq\sigPred$ of \emph{filter predicates}, where all filter predicates must be EDB predicates in $\aprogram$.
Outputs define which inferences are of interest to users, and are supported in many reasoners, including the ones in Table~\ref{tab_testeval}.
Filters are used in our algorithms to reduce inferences for non-output predicates, and would be defined by the reasoner
implementation: they should be easy to check and highly selective.
For example, the above predicates for $\predEquals$ are good filters,
whereas an EDB predicate that is stored in a large file on disk most likely is not. However, specific predicates
like ``$x$ is a string with letter \emph{e} in third position'' or pre-loaded data with
fast index structures can also be suitable filters.
Given a set of atoms $B$, we define $B_\sigPredFilter=\{p(\vec{t})\in B \mid p\in\sigPredFilter\}$
and $B_{\bar{\sigPredFilter}}=B\setminus B_\sigPredFilter$ as the conjunction of atoms with filter predicates and, respectively, atoms with non-filter predicates.

\paragraph*{Rules with generalised filters}
A \emph{rule with generalised filter expressions} has the form $H\leftarrow B_{\bar{\sigPredFilter}}\wedge G_\sigPredFilter$
with $H$ a head atom, $B_{\bar{\sigPredFilter}}$ a conjunction of non-filter atoms, and $G_\sigPredFilter\in\genFilters$ a
positive boolean combination of filter atoms, defined recursively:
\[\genFilters\Coloneqq p(\vec{t}): p\in\sigPredFilter, |\vec{t}|=\arity{p}, \vec{t} \subseteq \sigCons \cup \sigVar \mid (\genFilters\wedge\genFilters) \mid (\genFilters\vee \genFilters). \]
Positive boolean combinations of body atoms are normally syntactic sugar in Datalog: we can replace any (body) disjunction
$A[\vec{x}]\vee B[\vec{y}]$ over (possibly overlapping) sets of variables
$\vec{x}$ and $\vec{y}$ by a fresh atom $D[\vec{x}\cup\vec{y}]$, and add rules $D[\vec{x}\cup\vec{y}]\leftarrow A[\vec{x}]$ and
$D[\vec{x}\cup\vec{y}]\leftarrow B[\vec{y}]$. The fresh atom $D$ avoids the exponential blow-up that would occur if we would instead create
two copies of the rule, one with $A$ and one with $B$. 
With potentially infinite filter predicates, however, this syntactic transformation is not natural, since the auxiliary $D$ could be infinite, 
whereas it is easy for systems to evaluate nested expressions $G_\sigPredFilter$ in place.
Therefore, if not otherwise stated, all \emph{programs} below may include generalised filters. Their normal form is defined as for Datalog.

\section{Static Filtering: A General Algorithm}\label{sec_sf}

Next, we present a method for optimising Datalog programs by static rewriting. The optimised programs
have smaller least models that are nonetheless guaranteed to contain the same facts for output predicates $\sigPredOutput$.
A comparison with the work of Kifer and Lozinskii is given in Section~\ref{sec_rel_work}.

We consider a fixed program $\aprogram$ in normal form, a database $\adatabase$, filter predicates $\sigPredFilter$, and output predicates $\sigPredOutput$.
Facts for non-filter predicates in $\adatabase$ are irrelevant for static filtering.

\begin{example}\label{ex_bounded_reach}
As a running example, we consider a depth-bounded reachability check:
\begin{align}
r(x,y,n) &\leftarrow e(x,y)\wedge n\narrow\predEquals 0 \label{rule_reach_init}\\
r(x,z,m) &\leftarrow r(x,y,n)\wedge e(y,z)\wedge m\narrow\predEquals n{+}1 \label{rule_reach_step}\\
\textit{out}(y) &\leftarrow r(x,y,n)\wedge x\narrow\predEquals \const{a}\wedge n\narrow\leq 5 \label{rule_reach_out}
\end{align}
where $\textit{out}$ is the output predicate.
Notably, rule \eqref{rule_reach_step} can produce infinitely many inferences if the graph described
by $e$ is cyclic, but only finitely many nodes reachable from $\const{a}$ in ${\leq}5$ steps are
relevant for the output.
\end{example}

\paragraph*{The logic of filters}
For a given arity $k>0$,
let $\sigNumPos_k=\{\numpos{1},\ldots,\numpos{k}\}$ be a set of $k$ positional markers.
A \emph{filter atom} (for arity $k$) is an expression $\atom{f}=p(\numpos{m_1},\ldots,\numpos{m_\ell})$ where $p\in\sigPredFilter$ with $\ell = \arity{p}$ and $\numpos{m_i}\in\sigNumPos_k$ for $i=1,\ldots,\ell$.
The semantics of $\atom{f}$ is the relation $\atom{f}^\adatabase =\{ \vec{c}\in \sigCons^k\mid \tuple{c_{m_1},\ldots,c_{m_\ell}}\in p^\adatabase\}$.
Let $\sigPredFilter[k]$ be the set of all filter atoms of arity $k$ over $\sigPredFilter$.
The \emph{filter formulas} $\filterFormulas{k}$ are the positive boolean formulas over $\sigPredFilter[k]$:
\begin{align}
\filterFormulas{k} \Coloneqq \sigPredFilter[k]\mid \top \mid \bot \mid (\filterFormulas{k}\wedge\filterFormulas{k}) \mid (\filterFormulas{k}\vee \filterFormulas{k})
\end{align}
Their semantics is defined as expected: $\top^\adatabase=\sigCons^k$, $\bot^\adatabase=\emptyset$, $(F\wedge G)^\adatabase=F^\adatabase\cap G^\adatabase$, and 
$(F\vee G)^\adatabase=F^\adatabase\cup G^\adatabase$.
Given filter formulas $F,G\in\filterFormulas{k}$, we write $F\models G$ if $F^\adatabase\subseteq G^\adatabase$, and $F\equiv G$ if $F\models G$ and $G\models F$.
We can assume that $\bot$ and $\top$ are only used at the root level, using the usual simplifications:
$(\bot\wedge F)\mapsto\bot$, $(\top\wedge F)\mapsto F$, $(\bot\vee F)\mapsto F$, and $(\top\vee F)\mapsto\top$ (and their commutated versions).
A filter formula is \emph{simplified} if none of these rewritings applies to it.

Example~\ref{ex_bounded_reach} might use filter predicates $\blank\narrow\predEquals \const{a}$, $\blank\narrow\leq 5$, and $\blank\narrow\predEquals \blank{+}1$.
Since filter formulas do not allow constants, we assume distinct predicates for every pattern of constant use.
Implementations generally decide which filters to consider -- those that occur syntactically are required, but others can be added.

\paragraph*{Optimised filter computation}
For every IDB predicate $p\in\sigPred$, we seek a filter formula $\fnFilterExpr{p}\in\filterFormulas{\arity{p}}$, such that we only need to derive facts $p(\vec{c})$ if $\vec{c}\in\fnFilterExpr{p}^\adatabase$.
To obtain an algorithm, some operations have to be concretely implemented for the chosen filters $\sigPredFilter$ and every $k\geq 1$:
\begin{enumerate}
\item It must be possible to decide $F\models G$ for any $F,G\in\filterFormulas{k}$.
\item There is a canonical representation function $\reprFunc: \filterFormulas{k}\to\filterFormulas{k}$, such that $\repr{F}\equiv F$, and $F\equiv G$ implies $\repr{F}=\repr{G}$, for all $F,G\in\filterFormulas{k}$.
\end{enumerate}

For an atom $p(\vec{x})$ with arity $\arity{p} = k$, the mapping $\iota_{p(\vec{x})} \colon \numpos{i} \mapsto x_i$ maps the positional markers $\numpos{1}, \ldots, \numpos{k}$ to $\vec{x}$.
We extend $\iota_{p(\vec{x})}$ to filter formulas $F$, i.e., we obtain $\iota_{p(\vec{x})}(F)$ by replacing each positional marker $\numpos{i}$ with $\iota_{p(\vec{x})}(\numpos{i}) = x_i$.
For $r(x,y,n)$ of rule \eqref{rule_reach_step} and $F = \numpos{3} \leq 5$, e.g., we get
$\iota_{r(x,y,n)}(F) = \iota_{r(x,y,n)}(\numpos{3} \leq 5) = n \leq 5$.

\DontPrintSemicolon
\begin{algorithm}[tb] \caption{Static filter computation}\label{alg_filter_pushing}
	
 \KwIn{program $\aprogram$, output predicates $\sigPredOutput$}
 \KwOut{filter formulas $\fnFilterExpr{p}$ for IDB predicates $p$}

 \For{$p\in\sigPred$ where $p$ is an IDB predicate} {
    \leIf{$p\in\sigPredOutput$} {$\fnFilterExpr{p}\coloneqq\repr{\top}$} {\ghost{$\fnFilterExpr{p}\coloneqq\repr{\bot}$}} \label{line_fp_init}
 }
 \Repeat{all formulas $\fnFilterExpr{p}$ remain unchanged}  { \label{line_fp_mainloop}
    \For{$\arule\in\aprogram$ with $\arule = h(\vec{x})\leftarrow B_{\bar{\sigPredFilter}}\wedge G_\sigPredFilter$ \label{line_fp_ruleloop}} {
        \For{$b(\vec{y})\in B_{\bar{\sigPredFilter}}$ where $b$ is an IDB predicate \label{line_fp_bodyloop}} {
            $G \coloneqq \iota_{h(\vec{x})}(\fnFilterExpr{h})\wedge G_\sigPredFilter$\; \label{line_fp_matchfilter}
$M \coloneqq \bigwedge\{ F\in\filterFormulas{\arity{b}}\mid G\models \iota_{b(\vec{y})}(F)\}$\; \label{line_fp_bodyatomfilter}
            $\fnFilterExpr{b} \coloneqq \repr{\fnFilterExpr{b}\vee M}$\; \label{line_fp_filterunion}
        }
    }
 }
\end{algorithm}
Algorithm~\ref{alg_filter_pushing} describes the iterative computation of optimised filters.
Line \alglineref{line_fp_init} initialises filters to $\top$ for output predicates (``compute all facts''),
and to $\bot$ for all others.
The formulas are then generalised by looping over all rules (\alglineref{line_fp_ruleloop}) and their non-filter body atoms (\alglineref{line_fp_bodyloop}).
Matches to a rule can be restricted to those that satisfy both the head predicate's filter $\fnFilterExpr{h}$ and
the rule's own filter expression $G_\sigPredFilter$, which are combined into a filter formula $G\in\filterFormulas{|\vars{\arule}|}$ (\alglineref{line_fp_matchfilter});
we map $\fnFilterExpr{h}$ to the variables used in the rule.
To find the strongest filter formula $M$ for body atom $b(\vec{y})$, we take a conjunction of all filters $F\in\filterFormulas{\arity{b}}$ that
follow from $G$ when mapping $F$ to variables $\vec{y}$ as used in $b(\vec{y})$ (\alglineref{line_fp_bodyatomfilter}).
Finally, we generalise the current filter for $b$ by including the new $M$ disjunctively,
and taking the canonical representation (\alglineref{line_fp_filterunion}).

\begin{example}\label{ex_bounded_reach_sfcomp}
We apply Algorithm~\ref{alg_filter_pushing} to Example~\ref{ex_bounded_reach}, with a semantically minimised canonical representation
in disjunctive normal form. Initially, we have $\fnFilterExpr{\textit{out}}=\top$ and $\fnFilterExpr{r}=\bot$.
Processing body atom $r(x,y,n)$ of rule \eqref{rule_reach_out}, we get $G=\top\wedge x \narrow\predEquals \const{a} \wedge n \narrow \leq 5$.
Therefore, we have $G\models x \narrow\predEquals \const{a}=\iota_{r(x,y,n)}(\numpos{1}\narrow\predEquals \const{a})$ and $G\models n \narrow\leq 5=\iota_{r(x,y,n)}(\numpos{3}\narrow\leq 5)$,
and $M$ is equivalent to the conjunction $\numpos{1}\narrow\predEquals \const{a}\wedge \numpos{3}\narrow\leq 5$.
This is also the new filter condition $\fnFilterExpr{r}$.

In the next iteration, for body atom $r(x,y,n)$ in rule \eqref{rule_reach_step}, we get $G=x \narrow\predEquals \const{a}\wedge m \narrow\leq 5 \wedge m \narrow\predEquals n{+}1$.
The only relevant entailments are $G\models x \narrow\predEquals \const{a} =\iota_{r(x,y,n)}(\numpos{1}\narrow\predEquals \const{a})$
and $G\models n \narrow\leq 5 =\iota_{r(x,y,n)}(\numpos{3}\narrow\leq 5)$, so we obtain the same $M$ and $\fnFilterExpr{r}$ as before.
The algorithm then terminates.
\end{example}

\paragraph*{Correctness and use in optimisation}

Algorithm~\ref{alg_filter_pushing} is correct in the sense that rule applications can be safely restricted to applying rules only when
the conclusion $p(\vec{c})$ satisfies the computed filter formula $\fnFilterExpr{p}$, as there is no risk of changing derivations for
outputs. However, adding body atoms for $\fnFilterExpr{p}$ to all rules is often redundant.
The next definition characterises choices for equivalent filter conditions.

\begin{definition}\label{def_admissible}
Consider a rule $\arule = h(\vec{x})\leftarrow B_{\bar{\sigPredFilter}}\wedge G_\sigPredFilter \in \aprogram$.
Let $F_+ = \iota_{h(\vec{x})}(\fnFilterExpr{h})\wedge G_\sigPredFilter$ denote the
formula that combines the given filters $G_\sigPredFilter$ with the computed filter for $h(\vec{x})$, 
let $\sigPredIDB$ denote the IDB predicates of $P$, and let
$F_-=\bigwedge\{ \iota_{q(\vec{y})}(\fnFilterExpr{q}) \mid q(\vec{y})\in B_{\bar{\sigPredFilter}}, q \in \sigPredIDB\}$
denote the filter formula that combines the computed filters for IDB atoms in $B_{\bar{\sigPredFilter}}$.
Then a formula $\psi$ is \emph{admissible for $\arule$} if
\[ F_+ \models \psi \qquad\text{ and }\qquad \psi \wedge F_- \models F_+\;\text{.} \]
Then the rule $h(\vec{x})\leftarrow B_{\bar{\sigPredFilter}}\wedge \psi$
is an \emph{admissible rewriting} of $\arule$.\footnote{Technically, we simplify $\psi$: if $\psi=\bot$, the rewritten rule is deleted instead; if $\psi=\top$, then $\psi$ is omitted.}
An \emph{admissible rewriting} of $\aprogram$ is a set that contains an admissible rewriting of each rule of $\aprogram$.
\end{definition}

The conditions of Definition~\ref{def_admissible} ensure that admissible filters are equivalent to the
canonically extended body filter $F_+$ under the assumption that body atoms are also restricted to their computed filters $F_-$.
We obtain the following.

\begin{restatable}{theorem}{thmFilterPropagationCorrectness}\label{thmFilterPropagationCorrectness}
    If $\aprogram'$ is an admissible rewriting of $\aprogram$, and $p(\vec{c})$ is a 
    fact with $p\in\sigPredOutput$, then
    $\aprogram,\adatabase \models p(\vec{c})$ iff $\aprogram',\adatabase \models p(\vec{c})$.
\end{restatable}

\begin{example}\label{ex_bounded_reach_sfrewrite}
Using the filters computed in Example~\ref{ex_bounded_reach_sfcomp},
we can get this admissible rewriting:
\begin{align}
r(x,y,n) \leftarrow{}& e(x,y)\wedge n\narrow\predEquals 0\wedge x\narrow\predEquals \const{a} \label{rule_reach_init_opt}\\
r(x,z,m) \leftarrow{}& r(x,y,n)\wedge e(y,z) \wedge m\narrow\predEquals n{+}1 \wedge m\narrow\leq 5\label{rule_reach_step_opt}\\
\textit{out}(y) \leftarrow{}& r(x,y,n) \label{rule_reach_out_opt}
\end{align}
For rule \eqref{rule_reach_init_opt}, $n \narrow\predEquals 0\wedge x \narrow\predEquals \const{a}$ is admissible, as $n \narrow\predEquals 0\models n \narrow\leq 5$.
For rule \eqref{rule_reach_step_opt}, body atom $r(x,y,n)$ yields $F_-= x \narrow\predEquals \const{a}\wedge n \narrow\leq 5$,
while $F_+ = x \narrow\predEquals \const{a}\wedge m \narrow\leq 5\wedge m \predEquals n+1$; so $m \narrow\predEquals n{+}1 \wedge m\narrow\leq 5$ is indeed admissible.
Finally, for rule~\eqref{rule_reach_out_opt}, an empty ($\top$) filter is admissible, since $F_-$ already entails all necessary conditions.
Importantly, the final rewriting (and any other admissible rewriting) will always terminate, even when Example~\ref{ex_bounded_reach} fails to do so.
\end{example}

\begin{algorithm}[tb] \caption{Computing admissible filters}\label{alg_admin_filter}

 \KwIn{rule $\arule$, formulas $F_+$ and $F_-$ as in Definition~\ref{def_admissible}}
 \KwOut{formula $\psi$ admissible for $\arule$}

 $\psi\coloneqq F_+$\;
 \For{every occurrence $o$ of an atom in $\psi$} {
    \lIf{$F_-\wedge \psi[o\mapsto\top]\models \psi$} {$\psi \coloneqq \psi[o\mapsto\top]$} \label{line_adf_step}
 }
 \Return simplification of $\psi$
\end{algorithm}
Note that $F_+$ as given in Definition~\ref{def_admissible} is always admissible, but there can be much simpler expressions.
Algorithm~\ref{alg_admin_filter} shows a practical way to find good admissible filters, which would also
find the rewriting of Example~\ref{ex_bounded_reach_sfrewrite}.
An \emph{occurrence} $o$ of an atom in $\psi$ refers to a single position in $\psi$ where some atom is found (even if
the same atom occurs more than once), and $\psi[o\mapsto\top]$ denotes the result of replacing this occurrence by $\top$.
Each iteration in line \alglineref{line_adf_step} preserves admissibility of $\psi$, and the algorithm terminates 
after linearly many iterations.

\paragraph*{Cost and benefits of static filtering}

While the concrete performance impact of the rewriting depends on $\aprogram$ and $\adatabase$,
static filtering is guaranteed to only reduce but never increase the number of logical entailments.
The next result is a direct consequence of the proof of Theorem~\ref{thmFilterPropagationCorrectness},
where we showed that $\mathcal{M'} = \{ q(\vec{d}) \in \mathcal{M} \mid \satisfy{\vec{d}}{\fnFilterExpr{q}} \} \cup \adatabase$:

\begin{restatable}{theorem}{thmRewritingLinearWorstCase}\label{thmRewritingLinearWorstCase}
	Let $\aprogram'$ be an admissible rewriting of $\aprogram$, and let $\mathcal{M}$ (resp.\ $\mathcal{M}'$) be the model
	of $\adatabase$ and $\aprogram$ (resp.\ of $\adatabase$ and $\aprogram'$).
	Then $\mathcal{M}'\subseteq\mathcal{M}$, and every match of a rule $\arule'\in\aprogram'$ on $\mathcal{M}'$ is also
	a match of the corresponding rule $\arule\in\aprogram$ on $\mathcal{M}$.
\end{restatable}

In particular, every derivation (proof tree) over $\adatabase$ and $\aprogram'$ corresponds to a analogous derivation
over $\adatabase$ and $\aprogram$, and the bottom-up computation of $\mathcal{M}'$ requires at most as many rule applications
as the bottom-up computation of $\mathcal{M}$. The additional cost associated with the use of $\aprogram'$ is limited to the cost of checking,
for each rule match, if the rewritten filters rather than the original filters are satisfied.
This cost can be controlled by system designers through the choice of filter predicates to push.
The potential savings, on the other hand, can be in the order of $|\mathcal{M}|$.

The fact that static filtering preserves the structure of rules, programs, and derivations also improves
understandability by human users and compatibility with other optimisations, be it logical (e.g., magic sets) or operational (e.g., join-order optimisation).
The rewriting is even \emph{idempotent}, i.e., already optimised programs will not be modified further when applying
static filtering again. None of these advantages is common to all logic program optimisation methods,
a prominent counterexample being magic sets.
 \section{Lower and Upper Bounds for Termination}\label{sec_sf_steps}

\begin{table}[t]
	\centering
	\caption{Worst-case number of computing steps of Algorithm~\ref{alg_filter_pushing} in several scenarios, with details on upper and lower bounds depending on the filter size and head arity}
	\label{tab_sf_runtimes}
	\begin{tabular}{@{}cccl@{}cl@{}}
	\textbf{Filter arity} & \textbf{Filter size} & \textbf{Head arity} & \multicolumn{2}{c}{\textbf{\#Steps}} \\
	\hline \\[-2ex]
	\multirow{ 2}{*}{variable} & infinite & variable & $\leq$ & \multirow{ 2}{*}{doubly exponential} & {\footnotesize Thm~\ref{thm_sf_syntax_upper_bound}}\\
	 & doubly exponential & constant\ $\geq 2$ & $\geq$ & & {\footnotesize Prop~\ref{prop_sf_doubly_exp}}\\
	\hline \\[-2ex]
	\multirow{ 2}{*}{constant\ $\geq 2$} & infinite & \multirow{ 2}{*}{variable} & $\leq$ & \multirow{ 2}{*}{exponential} & {\footnotesize Thm~\ref{thm_sf_syntax_upper_bound}}\\
	 & polynomial &  &$\geq$ & & {\footnotesize Ex~\ref{ex_sf_exp_fconst}}\\
	\hline \\[-2ex]
	variable & infinite & \multirow{ 2}{*}{constant\ $=1$} & $\leq$ & \multirow{ 2}{*}{exponential} & {\footnotesize Thm~\ref{thm_sf_syntax_upper_bound}}\\
	constant\ $\geq 2$ & exponential &  & $\geq$ & & {\footnotesize Prop~\ref{prop_sf_exp_punary}}\\
	\hline \\[-2ex]
	variable & exponential & variable & $\leq$ & \multirow{ 2}{*}{exponential} & {\footnotesize Thm~\ref{thm_sf_semantics_upper_bound}}\\
	constant\ $\geq 2$ & polynomial & variable & $\geq$ & & {\footnotesize Ex~\ref{ex_sf_exp_fconst}}\\
	\hline \\[-2ex]
	variable & polynomial & constant\ $\geq 2$ & $\leq$ & polynomial  & {\footnotesize Thm~\ref{thm_sf_semantics_upper_bound}}
	\end{tabular}
\end{table}
Next, we analyse the time complexity of Algorithm~\ref{alg_filter_pushing} in terms of the
number of iterations.
For now, we abstract from the complexity of checking $\models$ and computing
representatives $\reprFunc$, which are discussed in more detail in Section~\ref{sec_fentailment}.
Our results are summarised in Table~\ref{tab_sf_runtimes}, depending on given bounds on predicate arity and size of filter relations,
where \emph{variable} arity and \emph{infinite} size are least restrictive.

\begin{example}\label{ex_kl_exp}
Kifer and Lozinskii \cite{KiferLozinskii:StaticFiltering90} find that static filtering is exponential, and they give the following example:
\begin{align}
	r(\vec{x},y) &\leftarrow p(\vec{x},y)\\
	r(\vec{x}_{i \leftrightharpoons j},y) &\leftarrow r(\vec{x},y) \qquad \text{for all $1\leq i<j\leq|\vec{x}|$}\label{eq_kl_exp_swap}\\
	\textit{out}(y) &\textstyle\leftarrow r(\vec{x},y) \land \bigwedge_{i=1}^{|\vec{x}|} x_i\predEquals \const{a}_i
\end{align}
where $\vec{x}_{i \leftrightharpoons j}$ denotes $\vec{x}$ with variables $x_i$ and $x_j$ swapped, and all $\const{a}_i$ are constants.
Rules \eqref{eq_kl_exp_swap} derive all (exponentially many) permutations of the tuples in $r$, and the filter computed for
$r$ must allow all permutations of $\vec{a}$. Rewriting to plain Datalog, we get rules
$r(\vec{x},y) \leftarrow p(\vec{x},y)\land \bigwedge_{i=1}^{|\vec{x}|} x_i\predEquals \const{a}_{\nu(i)}$ for exponentially many permutation functions $\nu$.
\end{example}

Example~\ref{ex_kl_exp} seems to establish an exponential lower bound for Algorithm~\ref{alg_filter_pushing}, but the exponential complexity
in this case stems merely from the representation of filter formulas. All permutations are obtained in linearly many pairwise swaps,
and Algorithm~\ref{alg_filter_pushing} therefore terminates after linearly many iterations of loop \alglineref{line_fp_mainloop}.
Restricting to the filter predicates in the example, the result does require an exponential filter formula,
but there are also filter logics that support polynomial representations, e.g., by simply introducing filter predicates for
statements like ``$\vec{b}$ is a permutation of $\vec{a}$''.
Example~\ref{ex_kl_exp} therefore is not illustrating exponential behaviour in general,
but we can find other examples that do:

\begin{example}\label{ex_sf_exp_fconst}
Let $p$ be a non-filter predicate of arity $\ell+1$. For readability, we do not normalise the rules:
\begin{align}
p(\vec{x},y) & \leftarrow e(\vec{x},y) \label{eq_ex_exp_in}\\
p(x_1,\ldots, x_i,1,0,\ldots,0,y) &\textstyle \leftarrow p(x_1,\ldots, x_i,0,1,\ldots,1,y)  \hspace{.8cm} \text{for all $i\in\{1,\ldots,\ell\}$} \label{eq_ex_exp_step}\\
\textit{out}(y) & \leftarrow p(1,\ldots,1,y) \label{eq_ex_exp_out}
\end{align}
Rule \eqref{eq_ex_exp_in} can be rewritten to require each variable in $\vec{x}$ to map to a constant from $\{0,1\}$,
which can be expressed in a compact generalised filter formula (with nested $\vee$).
However, the exponentially many admissible combinations of lists from $\{0,1\}^\ell$ are computed in Algorithm~\ref{alg_filter_pushing}
by iterating over the successor relation for binary numbers, realised by rules \eqref{eq_ex_exp_step}. This requires exponentially many steps.
Note that normalisation and static filtering in this case would only use filter predicates $\blank\narrow\predEquals\blank$ , $\blank\narrow\predEquals 0$,
and $\blank\narrow\predEquals 1$, as defined in Section~\ref{sec_prelims}.
\end{example}

In fact, as we will see below, Algorithm~\ref{alg_filter_pushing} may even require a doubly exponential number of iterations,
and still exponentially many for predicates of bounded arity.
The key to showing corresponding upper bounds is the following lemma.

\begin{restatable}{lemma}{lemmaChainlength}\label{lemma_chainlength}
Given $n_h$ distinct head predicates, if Algorithm~\ref{alg_filter_pushing} performs $\geq n_h\cdot s$ iterations of loop \alglineref{line_fp_mainloop}, then there is a
head predicate $p$, and a chain $F_1\models\cdots\models F_s$ of non-equivalent filter formulas $F_i\in\filterFormulas{\arity{p}}$.
\end{restatable}

In other words, (doubly) exponentially long runs require (doubly) exponentially ``deep'' filter logics.
The number of available filter formulas yields a first major upper bound.

\begin{restatable}{theorem}{thmSfSyntaxUpperBound}\label{thm_sf_syntax_upper_bound}
Let there be $n_\sigPredFilter$ filter predicates of arity $\leq a_\sigPredFilter$, and $n_h$ head predicates of arity $\leq a_h$.
Then Algorithm~\ref{alg_filter_pushing} terminates after at most $n_h\cdot 2^{n_\sigPredFilter\cdot {a_h}^{a_\sigPredFilter}}$ many iterations.
\end{restatable}

The doubly exponential worst case of Theorem~\ref{thm_sf_syntax_upper_bound} can be reached. As the theorem suggests, the
arity of head predicates can even be constant, as long as it is $>1$.

\begin{restatable}{proposition}{propExpSfDoublyExp}\label{prop_sf_doubly_exp}
	There are programs with filter arity $a_\sigPredFilter = \ell$ for which Algorithm~\ref{alg_filter_pushing} requires $2^{2^\ell}$ many iterations, even if the head arity is fixed and the size of the filter relations is at most doubly exponential.
\end{restatable}

Theorem~\ref{thm_sf_syntax_upper_bound} shows that the upper bound drops to single exponential if either (1) we impose a constant bound on the arity $a_\sigPredFilter$ of filter predicates, or (2) we fix the arity $a_h$ of the head predicates to $a_h=1$. Example~\ref{ex_sf_exp_fconst} showed exponential behaviour in case (1). We can also reach this upper bound in case (2).

\begin{restatable}{proposition}{propExpSfExpPunary}\label{prop_sf_exp_punary}
	There are programs with head arity $a_h=1$ and filter arity $a_\sigPredFilter = \ell$ for which Algorithm~\ref{alg_filter_pushing} requires $2^\ell$ many iterations, even if the size of the filter relations is at most exponential.
\end{restatable}

Proposition~\ref{prop_sf_doubly_exp} and \ref{prop_sf_exp_punary} consider filter relations that are of a size that is
proportional to the double and single exponential length of the runs, whereas Example~\ref{ex_sf_exp_fconst} only requires
polynomially sized filter relations. The following result clarifies how filter relation cardinality may affect upper bounds.

\begin{restatable}{theorem}{thmSfSemanticsUpperBound}\label{thm_sf_semantics_upper_bound}
Let there be $n_\sigPredFilter$ filter predicates that correspond to relations in $\adatabase$ of cardinality $\leq c_\sigPredFilter$,
and let there be $n_h$ head predicates of arity $\leq a_h$.
Then Algorithm~\ref{alg_filter_pushing} terminates after at most $n_h\cdot ((n_\sigPredFilter\cdot c_\sigPredFilter)^{a_h}+2)$ iterations.
\end{restatable}

Theorem~\ref{thm_sf_semantics_upper_bound} yields a polynomial upper bound for the case that
filter relations are polynomially bounded and non-filter predicates have a fixed arity.
However, many filters in existing systems correspond to infinite relations, so further 
approaches are required to make static filtering tractable in practice.

\section{Tractable Static Filtering}\label{sec_fentailment}

In this section, we further analyse the complexity of Algorithm~\ref{alg_filter_pushing},
and propose simplifications for making it tractable.
Section~\ref{sec_sf_steps} showed that runtimes can be prohibitive, even
without considering the cost of individual operations. For a full analysis, 
we must also analyse the cost of lines \alglineref{line_fp_bodyatomfilter} and \alglineref{line_fp_filterunion}.
Since the result of \alglineref{line_fp_filterunion} remains the same when 
replacing $M$ by any $M'\equiv M$, implementations can optimise \alglineref{line_fp_bodyatomfilter}
by computing a potentially smaller $M'$.
Nevertheless, Example~\ref{ex_kl_exp} shows a case where every such $M'$ still grows exponentially
during polynomially many iterations. Our proposed solution is to restrict
$\fnFilterExpr{p}$ to special forms of filter formulas that merely approximate those in Algorithm~\ref{alg_filter_pushing}.

First, however, we need to address the problem that even the entailment of individual filter formulas is generally undecidable, even for common filters.

\begin{restatable}{proposition}{propFilterLogUndec}\label{propFilterLogUndec}
If $\sigPredFilter$ contains predicates that can express arithmetic equalities
$x=y+z$ and $x=y\cdot z$ over natural numbers, and $x=n$ for all $n\in\mathbb{N}$,
then there is no algorithm that decides $F\models G$ for arbitrary $k>0$ and $F,G\in\filterFormulas{k}$.
\end{restatable}

It is well-known that arithmetic predicates are challenging to reason with, and
Datalog engines commonly restrict to \emph{safe} arithmetics, where all numeric variables are bound to
finite extensions of other predicates. However, this restriction does not 
simplify static filtering, where abstract filter formulas are considered without such concrete
bindings.

We therefore consider entailment relations that are sound but not necessarily complete,
except for the basic semantics of the propositional operators in filter formulas.

\begin{definition}\label{defApproxEntlaiment}
Every filter formula $F$ can be considered as a propositional logic formula over its filter atoms.
Let $\models_{\text{prop}}$ denote the usual propositional logic entailment relation over these 
formulas.
A binary relation $\approxmodels$ on filter formulas is an \emph{approximate entailment}
if ${\models_{\text{prop}}}\subseteq{\approxmodels}\subseteq{\models}$.
We write $F\approxequiv G$ if $F\approxmodels G$ and $G\approxmodels F$.
\end{definition}

We generalise Algorithm~\ref{alg_filter_pushing} to approximate entailments
by replacing any use of $\models$ by $\approxmodels$, and by allowing any representation function $\reprFunc: \filterFormulas{k}\to\filterFormulas{k}$ 
where $F\approxequiv\repr{F}$ and $F\approxequiv G$ implies $\repr{F}=\repr{G}$.
Definition~\ref{def_admissible} and Algorithm~\ref{alg_admin_filter} can likewise be generalised
by using filters computed by the approximate algorithm.
The main results of Sections~\ref{sec_sf} and \ref{sec_sf_steps} still hold in this generalised setting.

\begin{restatable}{lemma}{lemmaApproxSfCorrectTerm}\label{lemma_approx_sf_correct_term}
For any approximate entailment $\approxmodels$, Algorithm~\ref{alg_filter_pushing} terminates within the bounds 
of Theorem~\ref{thm_sf_syntax_upper_bound}.
If $\aprogram'$ is an admissible rewriting of $\aprogram$ based on $\approxmodels$, and $p(\vec{c})$ is a 
fact with $p\in\sigPredOutput$, then $\aprogram,\adatabase \models p(\vec{c})$ iff $\aprogram',\adatabase \models p(\vec{c})$.
\end{restatable}

If $\approxmodels$ is decidable, Algorithm~\ref{alg_filter_pushing} can be implemented, but may still be intractable.
In fact, deciding $\models_{\text{prop}}$ is still hard for $\coNP$. However, weakening $\approxmodels$ further to omit entailments of
$\models_{\text{prop}}$ may impair termination, since propositionally equivalent formulas may have distinct representatives.
To obtain a tractable procedure, we further modify Algorithm~\ref{alg_filter_pushing} so that
the formulas $\fnFilterExpr{p}$ can only be \emph{conjunctions} of filter atoms, $\top$, or $\bot$.
Lines \alglineref{line_fp_bodyatomfilter} and \alglineref{line_fp_filterunion} in the algorithm are now replaced by
\begin{align}
   \fnFilterExpr{b} := \bigwedge\{A\in\mathcal{A} \mid \iota_{b(\vec{y})}(\fnFilterExpr{b})\vee G\approxmodels\iota_{b(\vec{y})}(A) \}\label{eq_sf_fltconj}
\end{align}
where $\mathcal{A}=\{\bot\}\cup\sigPredFilter[\arity{b}]$, and we assume that conjunctions are represented as subsets of $\mathcal{A}$
with $\bigwedge\emptyset=\top$.
We refer to the modified algorithm as \emph{conjunctive approximate static filtering} (CASF). The complexity depends on the
choice of $\approxmodels$ and the size of $\sigPredFilter[\arity{b}]$, but the algorithm is correct in all cases.
This follows from the observation that the formulas $\fnFilterExpr{b}$ as computed in CASF are logical consequences (w.r.t.\ $\models$)
of those computed in Algorithm~\ref{alg_filter_pushing}, i.e., filters are more permissive.

\begin{restatable}{theorem}{thmCasfCorrect}\label{thm_casf_correct}
If $\aprogram'$ is an admissible rewriting of $\aprogram$ for the filter formulas computed in CASF, and $p(\vec{c})$ is a 
fact with $p\in\sigPredOutput$, then $\aprogram,\adatabase \models p(\vec{c})$ iff $\aprogram',\adatabase \models p(\vec{c})$.
\end{restatable}

For polynomial runtime, we need to restrict $\approxmodels$ so that the entailment in \eqref{eq_sf_fltconj} can be decided in \PTime.
We consider a finite set $\mathcal{T}$ of Datalog rules that use only filter predicates in head and body.
$\mathcal{T}$ is a \emph{Horn axiomatisation} for $\approxmodels$ if $F\approxmodels G$ holds for filter formulas $F,G\in\filterFormulas{k}$ exactly if
$G$ is a logical consequence of $F\cup\mathcal{T}$ (considered as a predicate logic theory over the domain of positional markers $\sigNumPos_k$).
Moreover, $\mathcal{T}$ is a \emph{linear axiomatisation} if rules in $\mathcal{T}$ have exactly one body atom.

\begin{restatable}{theorem}{thmCasfTractable}\label{thm_casf_tractable}
Let $\sigPredFilter$ be a set of filter predicates with bounded arity, and let
$\mathcal{T}$ be a (fixed) Horn approximation of $\approxmodels$.
Then CASF can be executed in polynomial time over a program $P$ in either of the following cases:
\begin{enumerate}
\item $\mathcal{T}$ is a linear approximation, or
\item the filter expressions $G_\sigPredFilter$ in $P$ do not contain $\vee$.
\end{enumerate}
\end{restatable}

Linear rules as in the first case in Theorem~\ref{thm_casf_tractable} suffice to model basic
hierarchies of filter conditions, but the power of Horn logic is required
to axiomatise typical binary filters such as order relations $\leq$ \cite{DBLP:journals/jacm/UllmanG88}.

\begin{example}\label{ex_theory_linear_order}
	For a set of natural numbers $N \subseteq \natnums$, consider the (possibly infinite) theory
	\begin{align}
		x\narrow\leq \const{c} &\leftarrow x\narrow\predEquals \const{c}  & &  \const{c} \in N\label{eq_ex_bounded_reach_axiom_init}\\
		x\narrow\leq \const{c} &\leftarrow y\narrow\leq \const{c} \land y\narrow\predEquals x{+}\const{d} & & \const{c}, \const{d} \in N \label{eq_ex_bounded_reach_axiom_prop}\\
		x\narrow\leq \const{c} &\leftarrow x\narrow\leq \const{d} & & \const{c}, \const{d} \in N, \const{c} > \const{d}\label{eq_ex_bounded_reach_axiom_monotone}
	\end{align}
	with rules instantiated for all values of $\const{c}$ and $d$ as specified.
The finite instantiation with $N={\{0,1,5\}}$ axiomatises all filter entailments necessary for Examples~\ref{ex_bounded_reach_sfcomp} and \ref{ex_bounded_reach_sfrewrite}.
	The relevant constants $N$ are syntactically given in the input filters.
\end{example}

We observe that the cases in Theorem~\ref{thm_casf_tractable} cannot be combined without loosing tractability:

\begin{restatable}{proposition}{propHornApproxCoNP}
There is a Horn approximation $\mathcal{T}$ of $\approxmodels$,
such that deciding $G\approxmodels A$ for a filter formula $G$ and an atom $A$
is hard for $\coNP$.
\end{restatable}

 \paragraph*{Tractable static filtering for real-world data}

Finally, we investigate the impact of (tractable) static filtering for reasoning over real-world data.
We implement conjunctive approximate static filtering (CASF) for linear orders and instantiations of the rules of Example~\ref{ex_theory_linear_order} for all natural numbers in a program $P$.
Moreover, we treat EDB predicates as filter predicates.
We rewrite programs for transitive closure over Wikidata properties via CASF, and we compare the runtime of the original programs with the rewritten ones as well as the runtime for static filtering.
In particular, we show that
\begin{enumerate}
	\item static filtering can improve the performance of modern rule systems by orders of magnitude,
	\item the simplifications in Section~\ref{sec_fentailment} are general enough to obtain these improvements, and
	\item the time necessary for applying tractable static filtering is negligible.
\end{enumerate}

As a typical examples for recursive programs, we use programs for computing the transitive closure of a predicate, and we add an output predicate $\textit{out}$ and a filter $x \predEquals \const{a}$ for a constant $\const{a}$.
We use a template for transitive closure (see Figure~\ref{fig_eval_templates}), which we instantiated for different EDB predicates $p(x,y)$.
We extract these predicates from Wikidata properties (see Figure~\ref{tab_eval_properties}). As rule systems, we considered Souffl\'{e} v2.5 \cite{Jordan+:Souffle16}, Nemo v0.8.1 \cite{Ivliev+:Nemo2024}, Clingo v5.8.0 \cite{Gebser+:clingo2019}, and DLV v2.1.2 \cite{Alviano+:DLV2:17}.
We have adopted the programs and inputs to the capabilities of the rule systems:
Souffl\'{e} and Nemo received facts as CSV files, while Gringo and DLV received them as a list of facts.
We applied a timeout at 5min.
Our measurements are performed on a regular notebook (Linux; AMD Ryzen 7 PRO 5850U; 16 GiB RAM).

\begin{figure}[t]
	\caption{Template programs for transitive closure over some EDB predicate $p \in \sigPred$;
		some filter predicate $x \predEquals a \in \sigPredFilter$ with constant $\const{a}$ is applied to compute the output predicate $\textit{out} \in \sigPredOutput$;
		original program (left) and rewritten program by tractable static filtering (right)}
	\label{fig_eval_templates}
	\vspace{-.5cm}
	\hspace*{\fill}
	\begin{subfigure}{.4\textwidth}
		\begin{align*}
			\textit{tc}(x,y) &\leftarrow p(x,y)\\
			\textit{tc}(x,z) &\leftarrow \textit{tc}(x,y)\wedge p(y,z)\\
			\textit{out}(y) &\leftarrow \textit{tc}(x,y)\wedge x \predEquals \const{a}
		\end{align*}
	\end{subfigure}
	\hspace*{\fill}
	\begin{subfigure}{.4\textwidth}
		\begin{align*}
			\textit{tc}(x,y) &\leftarrow p(x,y) \wedge x \predEquals \const{a} \\
			\textit{tc}(x,z) &\leftarrow \textit{tc}(x,y)\wedge p(y,z)\\
			\textit{out}(y) &\leftarrow \textit{tc}(x,y)
		\end{align*}
	\end{subfigure}
	\hspace*{\fill}
\end{figure}

\begin{figure}[t]
	\centering
	\caption{Table of used Wikidata properties used in the evaluation, i.e., the programs in Figure~\ref{fig_eval_templates} are instantiated with properties $p$ and entities $\const{a}$; \#Facts is the number of facts for property $p$ }\label{tab_eval_properties}
	\begin{tabular}{llrl}
		\multicolumn{1}{c}{\textbf{Property p}} & \multicolumn{1}{c}{\textbf{Property name}} & \multicolumn{1}{c}{\textbf{\#Facts}} & \multicolumn{1}{c}{\textbf{Entity $\const{a}$}} \\\hline \\[-2ex]
		P2652 & partnership with & 6,638 & Q180 (Wikimedia Foundation) \\
		P530 & diplomatic relation & 7,290 & Q33 (Finland) \\
		P1327 & partner in business or sport & 27,716 & Q1203 (John Lennon) \\
		P197 & adjacent station & 266,608 &  Q219867 (London King's Cross) \\
		P47 & shares border with & 927,553 & Q33 (Finland) \\
	\end{tabular}
\end{figure}

Figure~\ref{fig_eval_results} shows the runtimes for Programs~\ref{fig_eval_templates} instantiated with the properties of Table~\ref{tab_eval_properties}.
For each property, we run each system with the original program and the rewritten one.
In all cases, the rewritten programs require significantly less time – note the log-scale of the plot.
For the properties where the systems could finish for the original program, static filtering provides a speed-up in the order of magnitude (from $6.6$ times for P1203 and Clingo to $30$ times for P2652 and DLV).
Moreover, static filtering enables all systems to finish for huge properties with up to $1,000,000$ facts within seconds,
while all system reached the timeout for the original program there.
Finally, we observe that static filtering can be done in a few milliseconds, and it is independent of the number of facts for the underlying property.

\begin{figure}[t]
	\centering
	\caption{Runtimes (median of five runs) for transitive closure programs for different properties and rule systems;
		solid bars show runtime for original programs;
		hatched bars show runtime for rewritten programs;
		runtime of static filtering (black, solid lines); timeout (red, dotted lines) at 5min}\label{fig_eval_results}
	\vspace{-.2cm}
	\includegraphics[width=\textwidth]{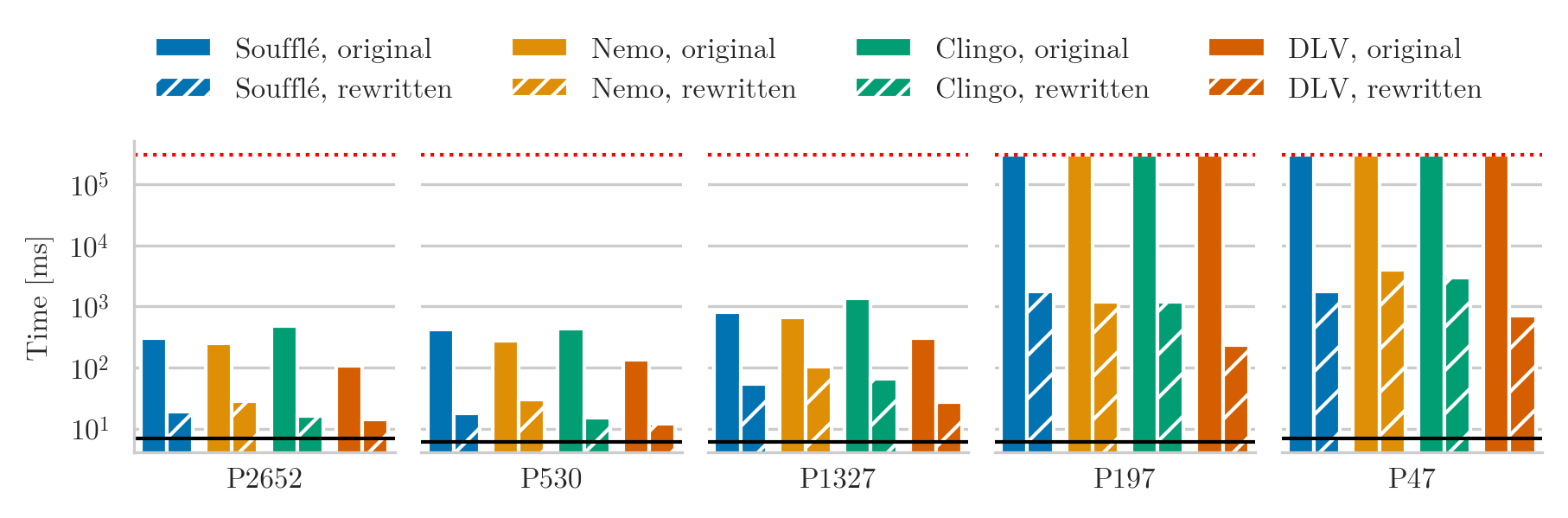}
\end{figure}

Our experiments show that (tractable) static filtering can improve the performance of rule systems significantly.
However, evaluating a static optimisation method empirically has its limitations.
Static optimization, by design, works on inputs that are not manually optimized yet.
While such programs are common in practice, there is no grounds for assuming anything about how common this is.
Existing public program collections are not representative here either,
since most of these published programs have been carefully optimised by experts
– exactly the type of manual work that static optimisations try to automate. 
\section{Incorporating Nonmonotonic Negation}\label{sec_negation}

In this section, we extend static filtering for rules with negation.
Datalog is often extended with \emph{stratified negation} \cite{Alice}.
The \emph{stable model semantics} as used in Answer Set Programming (ASP) \cite{DBLP:journals/cacm/BrewkaET11,DBLP:journals/ngc/GelfondL91} 
is a popular way to generalise this semantics to arbitrary rules with negation.
Both cases can benefit from static filtering; especially for ASP, where a polynomial reduction in
the grounding size can lead to an exponential performance advantage in solving.

Indeed, static optimisation in ASP is an important active topic of research.
Existing approaches include rewritings based on
tree decompositions \cite{DBLP:journals/fuin/BichlerMW20,DBLP:journals/tplp/CalimeriPZ19},
projection \cite{DBLP:conf/padl/HippenL19},
rule subsumption and shifting \cite{DBLP:conf/kr/EiterFTTW06}, and
magic sets \cite{DBLP:conf/iclp/CumboFGL04,DBLP:journals/tkde/Greco03} --
but we are not aware of any work that resembles static filtering.
Indeed, Table~\ref{tab_testeval} indicates that leading ASP engines do not
implement such optimisations even for basic filters of the form $\blank\narrow\predEquals \const{c}$.
The recent tool \emph{ngo}, maintained by the Clingo developers,
likewise implements many known optimisations, but no static filtering.\footnote{\url{https://potassco.org/ngo/ngo.html}}
When running $\emph{ngo}$ with all optimisations enabled on Example~\ref{ex_motivation},
it merely rewrites Rule~\ref{eq_testeval_out} to $\textit{out}(\const{b}) \leftarrow p(x_1,\ldots,x_\ell,\const{b})$,
which, expectably, does not have a notable impact on the runtime of any of the systems tested.

\paragraph*{Rules with negation}
A \emph{negated atom} is an expression $\naf p(\vec{t})$ with $p\in\sigPred$ and $\card{\vec{t}}=\arity{p}$.
A \emph{normal rule} $\arule$ is a formula $H \leftarrow B \land B^-$, where the head $H$ is an atom,
$B$ is a conjunction of atoms, and $B^-$ is a conjunction of negated non-filter atoms,
such that every $v\in\vars{\arule}$ occurs in some atom $p(\vec{x}) \in B$ (\emph{safety}).
Negated filter atoms are not needed: we can express them by introducing fresh filter predicates that
are interpreted by the complemented relation.
Analogously to Datalog, a \emph{normal rule with generalised filter expressions} has the form
$H\leftarrow B_{\bar{\sigPredFilter}}\wedge B_{\bar{\sigPredFilter}}^- \wedge G_\sigPredFilter$
with $H$ a head atom, $B_{\bar{\sigPredFilter}}$ a conjunction of non-filter atoms, and
$G_\sigPredFilter\in\genFilters$ a positive boolean combination of filter atoms.
In this section, all rules and programs may include negated atoms and generalised filters (we omit \emph{normal}).
The normal form without repeated variables per (negated or non-negated) atom is defined as for Datalog.

\paragraph*{Stable models}

For a program $P$ and database $\adatabase$, let $\ground{P} = \{ \arule\sigma \mid \arule \in P, \sigma \colon \sigVar \to \sigCons \}$ be its \emph{grounding}.
For a set of facts $\mathcal{A}$, the Datalog program $\ground{\aprogram}^{\mathcal{A}} = \{ H \leftarrow B \mid {H \leftarrow B \land B^-} \in \ground{\aprogram}, B^- \cap \mathcal{A} = \emptyset \}$ is the \emph{reduct}. $\mathcal{A}$ is a \emph{stable model} of $\aprogram$ and $\adatabase$ if $\mathcal{A}$ is the model of $\ground{\aprogram}^{\mathcal{A}}$ and $\adatabase$.
We write $\stablemods{\aprogram}{\adatabase}$ for the set of all such stable models.

\paragraph*{Stratified negation}
Even when using stable models, we check if programs are (partly) stratified, so as to tighten some filters.
Let $G_\aprogram$ be the graph with the IDB predicates of $\aprogram$ as its vertices,
and, for every rule $q(\vec{x}) \leftarrow B_{\bar{\sigPredFilter}}\wedge B_{\bar{\sigPredFilter}}^- \wedge G_\sigPredFilter \in \aprogram$,
a positive edge $p\to_+ q$ for all IDBs $p$ in $B_{\bar{\sigPredFilter}}$, and a negative edge $p \to_- q$ for all IDBs $p$ in $B_{\bar{\sigPredFilter}}^-$.
The \emph{stratifiable predicates} $\sigPredStrat$ are the vertices $p$ of $G_\aprogram$ such that there is no cycle $C$ in $G_\aprogram$
such that $C$ a negative edge and $p$ is reachable from $C$.

\paragraph*{Static filtering for stable models}
Importantly, facts that do not contribute directly to the output can still have an impact on stable models.
For example, consider a program $\aprogram$ with a stable model $\mathcal{A}$ that contains some $q(\vec{c})\in\mathcal{A}$;
then adding a rule $p(\vec{x}) \leftarrow q(\vec{x}) \land \naf p(\vec{x})$ for a fresh $p$ means that $\mathcal{A}$ is no
longer stable.
Hence, static filtering must not filter facts $p(\vec{c})$ relevant to some $\naf p(\vec{x})$.

We therefore define more general initial filters for negated predicates.
For a rule $\arule = h(\vec{x})\leftarrow B_{\bar{\sigPredFilter}} \wedge B_{\bar{\sigPredFilter}}^- \wedge G_\sigPredFilter$
with $b(\vec{y})\in B_{\bar{\sigPredFilter}}^-$,
let $M_{b(\vec{y})} = \bigwedge\{ F\in\filterFormulas{\arity{b}}\mid G_\sigPredFilter \models \iota_{b(\vec{y})}(F)\}$,
and let $N_\arule^p = \bigvee \{ M_{p(\vec{y})} \mid p(\vec{y})\in B_{\bar{\sigPredFilter}}^- \}$ with $\bigvee\emptyset = \bot$.
To obtain a procedure for a program $\aprogram$, we can now modify Algorithm~\ref{alg_filter_pushing} so that
the formulas $\fnFilterExpr{p}$ are initialised in line \alglineref{line_fp_init} with
\begin{align}
	\fnFilterExpr{p} \coloneqq \begin{cases}
		\repr{\top} & \textit{if } p\in\sigPredOutput\\
		\repr{\bigvee \{ N_\arule^p \mid \arule \in \aprogram \}} & \textit{if } p\notin\sigPredStrat\\
		\repr{\bot} & \textit{otherwise.}
	\end{cases}\label{eq_sf_init_asp}
\end{align}
Note that $\fnFilterExpr{p} \equiv \bot$ for $p \notin \sigPredStrat$ if $p$ never occurs in a negated atom.
Predicates $p \notin \sigPredStrat$ can be initialised with $\repr{\bot}$ as before, but we have to
consider them in the iterative generalisation: we modify line~\alglineref{line_fp_bodyloop} to loop over all $b(\vec{y})$ with
$b(\vec{y}) \in B_{\bar{\sigPredFilter}}$ or $\naf b(\vec{y}) \in B_{\bar{\sigPredFilter}}^-$ for IDB predicate $b$.

Our remaining definitions require only minimal adaptations.
For a rule $\arule = h(\vec{x})\leftarrow B_{\bar{\sigPredFilter}}\wedge B_{\bar{\sigPredFilter}}^-\wedge G_\sigPredFilter$, a filter formula $\psi$ is admissible for $\arule$ if $\psi$ is admissible for $h(\vec{x})\leftarrow B_{\bar{\sigPredFilter}}\wedge G_\sigPredFilter$, and $h(\vec{x})\leftarrow B_{\bar{\sigPredFilter}}\wedge B_{\bar{\sigPredFilter}}^-\wedge \psi$ is an admissible rewriting of $\arule$.
An \emph{admissible rewriting} of a program $\aprogram$ is a set that contains an admissible rewriting of each rule of $\aprogram$.
We can use Algorithm~\ref{alg_admin_filter} unchanged as a practical way to find good admissible filters.
Our main correctness results shows \emph{visible equivalence} \cite{DBLP:journals/jancl/Janhunen06} between $\aprogram$ and $\aprogram'$:

\begin{restatable}{theorem}{theoAspRewritingCorrect}
	If $\aprogram'$ is an admissible rewriting of program $\aprogram$ for database $\adatabase$, then
	$\mu \colon \mathcal{A}\mapsto\{p(\vec{c}) \in \mathcal{A} \mid \vec{c} \in \fnFilterExpr{p}^\adatabase\}$
	is a bijection between $\stablemods{\aprogram}{\adatabase}$ and $\stablemods{\aprogram'}{\adatabase}$.
\end{restatable}

In particular, since $\fnFilterExpr{p}\equiv\top$ for output predicates $p$, the restrictions of 
$\stablemods{\aprogram}{\adatabase}$ and $\stablemods{\aprogram'}{\adatabase}$ to facts over output predicates coincide.

One can easily incorporate the ideas of Section~\ref{sec_fentailment} to obtain a tractable optimisation procedure for normal logic programs.

 \section{Related work}\label{sec_rel_work}

\paragraph*{Comparison with the original algorithm}
The filter computation of Kifer and Lozinskii \cite{KiferLozinskii:StaticFiltering90}
can be seen as a special case of our approach for a fixed choice of filter predicates
(binary equalities and inequalities, possibly involving constants) and
representation of filter formulas (in disjunctive normal form).
Predicates such as $\blank\narrow\predEquals\blank{+}1$ in Example~\ref{ex_bounded_reach}
are not considered as filters. The significance of our generalisation is witnessed by
exponential increases in complexity (Section~\ref{sec_sf_steps}), but also by the ability to
introduce tractable simplifications (Section~\ref{sec_fentailment}).
The general notion of \emph{admissibility} and Algorithm~\ref{alg_admin_filter} are also new.

Kifer and Lozinskii further include a similar method to propagate projections and remove unused predicate parameters.
This rewriting is simpler than filter propagation.
We have nothing to add to it but note that it is particularly
effective if static filtering is applied first. 
\begin{example}\label{ex_bounded_reach_sfrewrite_projected}
	Propagation of projections does not lead to any simplification for Example~\ref{ex_bounded_reach},
	but leads to the following rules with reduced arities for the rewriting of Example~\ref{ex_bounded_reach_sfrewrite}:
	\begin{align}
		r'(y,n) &\leftarrow e(x,y)\wedge n\narrow\predEquals 0\wedge x\narrow\predEquals \const{a} \label{rule_reach_init_opt_proj}\\
		r'(z,m) &\leftarrow r'(y,n)\wedge e(y,z)\wedge m\narrow\predEquals n{+}1 \wedge m\narrow\leq 5\label{rule_reach_step_opt_proj}\\
		\textit{out}(y) &\leftarrow r'(y,n) \label{rule_reach_out_opt_proj}
	\end{align}
	Instead of computing the distance of quadratically many pairs $x$ and $y$, only the distance from $a$ to linearly many $y$ is needed.
	In general, reducing predicate arities can have big performance advantages, as it may simplify data structures and execution plans.
\end{example}

\paragraph*{Comparison with static optimization techniques}

Kifer and Lozinskii have already compared their special case of static filtering to existing methods including magic sets \cite{DBLP:conf/pods/BancilhonMSU86,DBLP:conf/sigmod/MumickFPR90,DBLP:conf/sigmod/MumickP94},
and they have have already observed that the approaches by Walker \cite{walker1981syllog} and Gardarin et al.\@ \cite{DBLP:conf/db-workshops/GardarinMS85} are similar, yet less general.
Subsequently, Chang et al.\@ extended the original method of Kifer and Lozinskii
to programs with stratified negation, which is less general than our extension to ASP.

\emph{Constraint pushing} \cite{kemp1989propagating,DBLP:journals/jlp/SrivastavaR93} considers rules with constraint atoms, with a propagation scheme
similar to Algorithm~\ref{alg_filter_pushing}. However, the method might not terminate, as there can be infinitely many constraint atoms.

Zaniolo et al.\@ introduce \emph{pre-mappability (PreM)} as a sufficient condition
for pushing filters into recursive rules, identify some classes of PreM filters,
and use such filters to rewrite recursive programs with aggregates \cite{DBLP:journals/tplp/ZanioloYDSCI17}.
Our Algorithm~\ref{alg_filter_pushing} can be seen as a systematic method for producing pre-mappable filters 
($\fnFilterExpr{p}$ are pre-mappable). 

The \emph{FGH-rule} by Wang et al.\@ \cite{DBLP:conf/sigmod/WangK0PS22} defines a sufficient condition for rewriting
a program using given output predicates (or queries), and our admissile rewritings satisfy this condition
(which is true for any rewriting that produces the same output facts).
Hence, static filtering offers a tractable method to find FGH-rule conforming rewritings.

In general, static filtering promises to work well with some rewriting techniques such as projection propagation and pre-mappability \cite{DBLP:journals/tplp/ZanioloYDSCI17}, and it is unlikely that it interferes negatively with static optimisations techniques, since static filtering preserves the program structure.

\paragraph*{Comparison to magic sets and demand transformation}

Magic sets \cite{DBLP:conf/pods/BancilhonMSU86,DBLP:conf/sigmod/MumickFPR90,DBLP:conf/sigmod/MumickP94} and the closely related \emph{demand transformation} \cite{DBLP:conf/ppdp/TekleL10,DBLP:journals/corr/abs-1909-08246} are static optimisation methods that also aim at reducing inferences by making some rule bodies more selective,
and which may seem similar to static filtering on an intuitive level.
However, in almost all cases, their outputs are very different from ours, for the following reasons:
\begin{enumerate}
	\item Static filtering preserves the number of rules and the structure of their non-filter atoms.
		Magic sets and demand transformation always increase the number of rules and add rules that contain partial body joins.
	\item Static filtering cannot optimise programs that do not contain filter predicates.
		Demand transformation can be used on purely abstract programs.
	\item Static filtering uses symbolic reasoning over filters to simplify and summarise expressions, so rewritten rules may contain new derived filters.
		Magic sets and demand transformation foresee no mechanism for integrating any symbolic knowledge about existing EDB predicates, so rewritten rules are always based on copies of EDB atoms that are syntactic parts of the program.
	\item Static filtering is idempotent (optimised programs are not rewritten further).
		The transformation by magic sets and demand transformation always changes the program, even if applied to its own output.
	\item Static filtering makes use of recursive rules for recursively generalising filter expressions. 
		Demand transformation in turn supports rewritings of IDB predicates that are defined by recursive rules.
	\item Static filtering includes a simplification step that removes filters that have become redundant after pushing (via \emph{admissibility}, Def.~\ref{def_admissible}). 
		Magic sets and demand transformation have no mechanism to detect possible simplifications.
\end{enumerate}
Therefore, we do not see any general principle to obtain the benefits of static filtering from magic sets or demand transformation,
even in special cases or with further adjustments.
Rather, the methods are complementary and can be used together, where static filtering should go first since any
reduction in body atoms can reduce the cost of the other transformation.
 \section{Conclusions}\label{sec_conclusions}

``It is folk wisdom that the right concepts are rediscovered several times'' is how
Kifer and Lozinskii started their conclusions \cite{KiferLozinskii:StaticFiltering90}.
In our work, we have revisited and generalised their original approach to static filtering,
presented a tractable simplification, and shown how to extend its use to Datalog with stratified negation and ASP.
Our framework lets implementers control which filters to push and which logical interactions to consider, 
and thus to avoid cases where the optimisation might cost more than it saves.
Since static filtering also preserves rule and proof structures, it plays well with other
optimisations and may even boost them (as in Example~\ref{ex_bounded_reach_sfrewrite_projected}).
It truly seems the ``right concept'' for many uses, in particular for data-oriented applications where logic programs
play the role of queries over potentially large datasets.
Thanks to the generality of our framework and the presented simplifications,
we are confident that any rule-based system can find a sweet spot where a small amount of effort can yield
decisive advantages at least in some cases.

Besides speeding up today's programs, however, we should also look for new uses ahead.
One is modularisation, since re-usable logic programming libraries
cannot be optimised manually for (yet unknown) usage contexts.
Another is termination, for arithmetic features as shown in Example~\ref{ex_bounded_reach_sfrewrite}, 
but also for rule languages with function symbols or existential quantifiers. 
These and other directions merit further research.

\bibliography{references}
	
	\clearpage
	\appendix

\section{Proofs for Section~\ref{sec_sf}}

\thmFilterPropagationCorrectness*
\begin{proof}
	Let $\mathcal{M}$ be the model of $\aprogram$ and $\adatabase$ and let $\mathcal{M'} = \{ q(\vec{d}) \in \mathcal{M} \mid \satisfy{\vec{d}}{\fnFilterExpr{q}} \} \cup \adatabase$.
	We show that $\mathcal{M'}$ is the model of $\aprogram'$ and $\adatabase$.
	
	By definition, $\adatabase \subseteq \mathcal{M'}$.
	
	$\mathcal{M'}$ is closed under $\aprogram'$:
	Let $\arule' = h(\vec{x}) \leftarrow \shBodyNormal \land \psi \in \aprogram'$ be the admissible rewriting of $\arule = h(\vec{x}) \leftarrow \shBodyNormal \land \shBodyFilterG \in \aprogram$.
Let $\sigma$ be a mapping such that $\shBodyNormal\sigma \subseteq \mathcal{M'}$ and $\psi\sigma \subseteq \adatabase$.
	For $b(\vec{y}) \in \shBodyNormal$ with IDB predicate $b$, we have $b(\sigma(\vec{y})) \in \mathcal{M'} \setminus \adatabase$ and $\satisfy{\sigma(\vec{y})}{\fnFilterExpr{b}}$.
	Let $F_-$ and $F_+$ be defined as in Def.~\ref{def_admissible}.
	We have $(F_- \land \psi)\sigma \subseteq \adatabase$.
	By admissibility of $\arule'$, we get $F_+\sigma \subseteq \adatabase$.
	In particular, $\satisfy{\sigma(\vec{x})}{\fnFilterExpr{h}}$ and $\shBodyFilterG\sigma \subseteq \adatabase$.
	Hence, $(\shBodyNormal \land \shBodyFilterG)\sigma \subseteq \mathcal{M'} \subseteq \mathcal{M}$ and $h(\sigma(\vec{x})) \in \mathcal{M}$, as $\mathcal{M}$ is the model of $\aprogram$ and $\adatabase$.
	Since $\satisfy{\sigma(\vec{x})}{\fnFilterExpr{h}}$, $h(\sigma(\vec{x})) \in \mathcal{M}'$, i.e., $\mathcal{M'}$ is closed under rule applications in $\aprogram'$.
	
	Minimality of $\mathcal{M'}$:
	For a contradiction, suppose there is a non-empty set $\mathcal{N}\subseteq\mathcal{M'}$ such that
	$\mathcal{M'}_- = \mathcal{M'} \setminus \mathcal{N}$ is also a model of $P'$ and $\adatabase$.
	In particular, $\adatabase\subseteq\mathcal{M'}_-$, so $\mathcal{N} \cap \adatabase = \emptyset$.
	Let $\mathcal{M}_- = \mathcal{M} \setminus \mathcal{N}$.
	Since $\mathcal{M}$ is the model of $\aprogram$ and $\adatabase$,
	there is a rule $\arule = h(\vec{x}) \leftarrow \shBodyNormal \land \shBodyFilterG \in \aprogram$ that is not satisfied by $\mathcal{M}_-$,
	i.e., there is a mapping $\sigma$ such that
	$h(\sigma(\vec{x})) \in \mathcal{N}$ and $(\shBodyNormal \land \shBodyFilterG)\sigma \subseteq \mathcal{M}\setminus \mathcal{N}$.
Since $h(\sigma(\vec{x})) \in \mathcal{N} \subseteq \mathcal{M'}\setminus \adatabase$, we have $\satisfy{\sigma(\vec{x})}{\fnFilterExpr{h}}$
	and therefore $\iota_{h(\vec{x})}(\fnFilterExpr{h})\sigma \subseteq \adatabase$ (A).
	Moreover, $\shBodyFilterG\sigma \subseteq \adatabase$ since $\arule$ is applicable in $M_-$ (B).
	
	Now consider an arbitrary $b(\vec{y}) \in \shBodyNormal$, and let $G$ and $M$ be defined as in Algorithm~\ref{alg_filter_pushing}.
	Then $G\sigma \subseteq \adatabase$ (by line \eqref{line_fp_matchfilter} with (A) and (B));
	$G \models \iota_{b(\vec{y})}(M)$ (by line \eqref{line_fp_bodyatomfilter}); and
	$\iota_{b(\vec{y})}(M) \models \iota_{b(\vec{y})}(\fnFilterExpr{b})$ (by line \eqref{line_fp_filterunion});
	therefore we have $\iota_{b(\vec{y})}(\fnFilterExpr{b})\sigma \subseteq \adatabase$, i.e., $\satisfy{\sigma(\vec{y})}{\fnFilterExpr{b}}$.
	Hence, $b(\sigma(\vec{y})) \in \mathcal{M}_-'$, and since $b(\vec{y})$ was arbitrary $\shBodyNormal\sigma \subseteq \mathcal{M}_-'$.
	Let $F_+$ be defined as in Def.~\ref{def_admissible}.
	Since $F_+\sigma \subseteq \adatabase$ and $\arule'$ is an admissible rewriting, $\psi\sigma \subseteq \adatabase$.
	Therefore, $\mathcal{M}_-'$ is not closed under application of $\arule'$ and $\sigma$, which yields the required contradiction.
\end{proof}

\section{Proofs for Section~\ref{sec_sf_steps}}

\propExpSfDoublyExp*

\begin{proof}
	For $\ell\geq 1$, we assume that $\sigCons$ contains constants $0,1\in\sigCons$ and strings in $\{0,1\}^{2^\ell}\subseteq\sigCons$ (or constants that can be interpreted as such, e.g., by taking the binary expansion of natural numbers).
	For $d\in\{0,1\}$, we define several filter predicates of the indicated arities as follows:
	\begin{itemize}
		\item $\adatabase\models\textit{max}(s)$ if $s$ is the string of $2^\ell$ repetitions of $1$, 
		\item $\adatabase\models\textit{is}_d(p_1,\ldots,p_\ell,s)$ if $s$ is a string that has
		$d$ at position with binary encoding $p_1\cdots p_n$,
		\item $\adatabase\models\textit{last}_d(p_1,\ldots,p_\ell, s)$ if the last occurrence of 
		$d$ in $s$ is at position with binary encoding $p_1\cdots p_n$,
		\item $\adatabase\models\textit{same}(p_1,\ldots,p_\ell, s, s')$ if strings $s$ and $s'$ agree on all symbols 
		at positions strictly smaller than $p_1\cdots p_n$.
	\end{itemize}
Consider the following program with variables $a,b,s,t,\vec{x}$:
\begin{align}
		n(s,a,b) & \leftarrow e(s)\wedge d(a)\wedge d(b) \label{eq_ex_doublyexp_in}\\
		n(s,a,b) &\textstyle \leftarrow n(t,a,b)\wedge\bigwedge_{i=1}^\ell d(x_i)\wedge{} \textit{last}_0(\vec{x},t)\wedge \textit{last}_1(\vec{x},s)\wedge \textit{same}(\vec{x},s,t) \label{eq_ex_doublyexp_step}\\
		\textit{out}(s) & \leftarrow n(s,a,b)\wedge \textit{max}(s)\wedge a\approx 0\wedge b\approx 1 \label{eq_ex_doublyexp_out}
	\end{align}
In the first iteration of Algorithm~\ref{alg_filter_pushing}, rule \eqref{eq_ex_doublyexp_out} yields
	$\fnFilterExpr{n}=F_0\in\filterFormulas{3}$, where $F_0$ is a conjunction that includes $\numpos{2}\narrow\approx 0$, $\numpos{3}\narrow\approx 1$, and
	exponentially many atoms of the form $\textit{is}_1(\vec{p},\numpos{1})$, where $\vec{p}$ is a list of $\ell$ positional markers ${\in}\{\numpos{2},\numpos{3}\}$.
	Note how the second and third parameter of $n$ is required for these atoms to be expressible.
	
	In the next iteration, for rule \eqref{eq_ex_doublyexp_step} we consider the conjunction of $F_0$ with the filter formula given in the body.
Since $F_0$ completely determines $\numpos{1}=\tonumpos{s}$ to represent the string $1\cdots 1$,
	the atom $\textit{last}_1(\tonumpos{\vec{x}},\tonumpos{s})=\textit{last}_1(\numpos{5},\ldots,\numpos{\ell{+}4},\numpos{1})$ entails
	$\numpos{i}\narrow\approx 1$ for all $\numpos{i}\in\{\numpos{5},\ldots,\numpos{\ell{+}4}\}$.
	Therefore, from $\textit{last}_0(\vec{x},t)$ and $\textit{same}(\vec{x},s,t)$, we conclude
	$\textit{is}_d$ atoms that express that $t$ is a string of the form $1\cdots 10$. This information can be projected to update $\fnFilterExpr{n}$, which now admits two strings.
	Continuing this process, $\fnFilterExpr{n}$ eventually becomes equivalent to a disjunction over conjunctions characterising all strings over $\{0,1\}$ of length $2^\ell$.
	There are $2^{2^\ell}$ such strings, so Algorithm~\ref{alg_filter_pushing} requires this many iterations.
\end{proof}

\propExpSfExpPunary*
\begin{proof}
	For $\ell\geq q$, we assume that $\sigCons$ contains constants $0,1\in\sigCons$ and strings in $\{0,1\}^\ell\subseteq\sigCons$ (or constants that can be interpreted as such, e.g., by taking the binary expansion of natural numbers).
	We define a unary filter predicate $\predOddK{k}$ such that $\adatabase\models\predOddK{k}(s)$ if $s$ is a string that contains $1$ at position $k$, and
	similarly for predicate $\predEvenK{k}$.
	Moreover, let $\adatabase\models\textit{same}_k(s, s')$ if strings $s$ and $s'$ agree on all symbols 
	at positions strictly smaller than $k$.
\begin{align}
		n(s) & \leftarrow e(s) \label{eq_ex_expunary_in}\\
		n(s) &\textstyle \leftarrow n(t)\wedge\predOddK{k}(s)\wedge\bigwedge_{i=k+1}^{\ell} \predEvenK{i}(s) \wedge \predEvenK{k}(t)\wedge\bigwedge_{i=k+1}^{\ell} \predOddK{i}(t)\wedge \textit{same}_k(s,t) \label{eq_ex_expunary_step}\\
		\textit{out}(s) & \textstyle\leftarrow n(s)\wedge \bigwedge_{i=1}^{\ell} \predOddK{i}(s) \label{eq_ex_expunary_out}
	\end{align}
where rule \eqref{eq_ex_expunary_step} is instantiated for all $k\in\{1,\ldots,\ell\}$.
	Algorithm~\ref{alg_filter_pushing} proceeds step by step, as in the proof for Proposition~\ref{prop_sf_doubly_exp}, but over
	merely exponentially many strings, which can be addressed by the polynomially many filter predicates.
\end{proof}

\section{Proofs for Section~\ref{sec_fentailment}}

\propFilterLogUndec*
\begin{proof}
	Let $\natnums \subseteq \sigCons$, 
	and let $(\blank\narrow\approx d)^{\adatabase}=\{d\}$ with $d \in \natnums$, 
	$(\blank\narrow\approx \blank + \blank)^\adatabase=\{\tuple{a,b,c} \in \natnums^3 \mid a = b + c \}$, 
	and $(\blank\narrow\approx \blank \cdot \blank)^\adatabase=\{\tuple{a,b,c} \in \natnums^3 \mid a = b \cdot c \}$ be predicates in $\sigPredFilter$.
	Let $f = g$ be a Diophantine equation with polynomials $f = f(\vec{x})$ and $g = g(\vec{y})$ with coefficients in $\natnums$.
	
	For a polynomial $p(\vec{v})$, the set $T_{p(\vec{v})}$ of arithmetic terms is the smallest set such that
	(i) $p(\vec{v}) \in T_{p(\vec{v})}$,
	(ii) if $(t \cdot u) \in T_{p(\vec{v})}$, then $t, u \in T_{p(\vec{v})}$, and
	(iii) if $(t + u) \in T_{p(\vec{v})}$, then $t, u \in T_{p(\vec{v})}$.
	Let $T = T_{f} \cup T_{g} \cup \{ 0 \}$, and
	let $\sigma \colon T \to \sigNumPos_{\card{T}}$ be a bijection.
	We define a translation $r \colon T \to \filterFormulas{\card{T}}$ inductively:
	\begin{align*}
		r(t) = \begin{cases}
			\sigma(t) \narrow\approx t & \textit{if } t \in \natnums,\\
			\top & \textit{if } t \in \vec{x} \cup \vec{y},\\
			(t \narrow\approx u \narrow+ v)\sigma \land  r(u) \land r(v) & \textit{if } t = (u + v),\\
			(t \narrow\approx u \narrow\cdot v)\sigma \land r(u) \land r(v) & \textit{if } t = (u \cdot v)
		\end{cases}
	\end{align*}
	
	We show via structural induction that,
	for term $t(\vec{v})$ and $\vec{n} \in \natnums^{\card{T}}$, we have $\vec{n} \in r(t)^\adatabase$ iff, for all $u = u(\vec{z}) \in T_t$, we have $n_{\sigma(u)} = u(n_{\sigma(z_1)},\ldots,n_{\sigma(z_{\card{\vec{z}}})})$:
	\begin{itemize}
		\item $t \in \natnums$: $r(t)^\adatabase = (\sigma(t) \narrow\approx t)^\adatabase = \{ \vec{n} \in \natnums^{\card{T}} \mid n_{\sigma(t)} \in  (\blank \narrow\approx t)^\adatabase \} = \{ \vec{n} \in \natnums^{\card{T}} \mid n_{\sigma(t)} = t \}$.
		\item $t = v \in \vec{v}$: By $t = t(v) = v$, we have that $n_{\sigma(t)} = t(n_{\sigma(t)})$ is true for all $\vec{n}$, i.e., $r(t) \equiv \top$.
		\item $t(\vec{z}) = (u(\vec{x}) + v(\vec{y}))$: Let $\vec{n} \in \natnums^{\card{T}}$.
			By induction, $\vec{n} \in (r(u) \land r(v))^\adatabase$ iff for all $w = w(\vec{z}) \in T_t \setminus \{ t \}$, we have $n_{\sigma(w)} = w(n_{\sigma(z_1)},\ldots,n_{\sigma(z_{\card{\vec{z}}})})$.
			Since we have $t(n_{\sigma(z_1)},\ldots,n_{\sigma(z_{\card{\vec{z}}})}) = u(n_{\sigma(x_1)},\ldots,n_{\sigma(x_{\card{\vec{x}}})}) + v(n_{\sigma(y_1)},\ldots,n_{\sigma(y_{\card{\vec{y}}})}) = n_{\sigma(u)} + n_{\sigma(v)}$, we have $\vec{n} \in ((t \narrow\approx u + v)\sigma)^\adatabase$ iff $n_{\sigma(t)} = n_{\sigma(u)} + n_{\sigma(v)}$ iff $n_{\sigma(t)} = t(n_{\sigma(z_1)},\ldots,n_{\sigma(z_{\card{\vec{z}}})})$.
		\item $t(\vec{z}) = (u(\vec{x}) \cdot v(\vec{y}))$: Analogously to the previous case (with $\cdot$ instead of $+$).
	\end{itemize}
	
	Let $F = r(f) \land r(g) \land (f \narrow\approx g + 0)\sigma \land \sigma(0)\narrow\approx0 $ and $G = \bot$.
	If $\vec{n} \in F$, then $\mu \colon \vec{x} \cup \vec{y} \to \natnums \colon z \mapsto n_{\sigma(z)}$ is a solution for $f = g$. 
	If $\mu \colon \vec{x} \cup \vec{y} \to \natnums$ is a solution for $f = g$, then $\vec{n} \in F^\adatabase$ with $n_i = \mu(t)$ for $t$ with $\sigma(t) = \numpos{i}$.
	Hence, $F \models G$ iff $f = g$ has no solution over $\natnums$.
\end{proof}

\lemmaApproxSfCorrectTerm*
\begin{proof}
	Termination within the bounds of Theorem~\ref{thm_sf_syntax_upper_bound} can be shown exactly as Theorem~\ref{thm_sf_syntax_upper_bound} itself.
	
	Since ${\approxmodels} \subseteq {\models}$, i.e., for formulas $F, G$, we have $G \approxmodels F$ implies $G \models F$,
	any admissible rewriting for $\approxmodels$ is admissible for $\models$ and
	the second part of this lemma can be shown almost exactly as Theorem~\ref{thmFilterPropagationCorrectness}.
	Let $\mathcal{M}$ be the model of $\aprogram$ and $\adatabase$ and let $\mathcal{M'} = \{ q(\vec{d}) \in \mathcal{M} \mid \satisfy{\vec{d}}{\fnFilterExpr{q}} \} \cup \adatabase$.
	We show that $\mathcal{M'}$ is the model of $\aprogram'$ and $\adatabase$ as for Theorem~\ref{thmFilterPropagationCorrectness}:
	\begin{itemize}
		\item By definition, $\adatabase \subseteq \mathcal{M'}$.
		\item Showing that $\mathcal{M'}$ is closed under $\aprogram'$ is independent of the approximation $\approx$,
			and we can show it as before.
		\item Showing the minimality of $\mathcal{M'}$ requires that, for $b(\vec{y}) \in \shBodyNormal$ and $G$ and $M$ defined as in Algorithm~\ref{alg_filter_pushing}, $G \models \iota_{b(\vec{y})}(M)$ and $M \models \fnFilterExpr{b}$.
		The approximation based on $\approxmodels$ preserves these properties, and the argumentation for minimality of $\mathcal{M'}$ in the proof of Theorem~\ref{thmFilterPropagationCorrectness} remains valid.\qedhere
	\end{itemize}
\end{proof}

\thmCasfCorrect*
\begin{proof}
	We observe that allowing only conjunctions in filter formulas leads to more general filters,
	since, for $\atom{b} = b(\vec{y}) \in \shBodyNormal$ and $G$ and $M$ defined as in Algorithm~\ref{alg_filter_pushing}, we have
	\begin{align}
		\repr{\fnFilterExpr{b}\shortvee M} \,{\models}\, \bigwedge\{A\,{\in}\,\mathcal{A} \,{\mid}\, \iota_{\atom{b}}(\fnFilterExpr{b})\shortvee G\,{\approxmodels}\,\iota_{\atom{b}}(A) \}\label{eq_conjunction_generalises}
	\end{align}
	where $\mathcal{A}=\{\bot\}\cup\sigPredFilter[\arity{b}]$, and we assume that conjunctions are represented as subsets of $\mathcal{A}$
	with $\bigwedge\emptyset=\top$.
	
	Hence, the theorem can be shown almost exactly as Theorem~\ref{thmFilterPropagationCorrectness}.
	Let $\mathcal{M}$ be the model of $\aprogram$ and $\adatabase$ and let $\mathcal{M'} = \{ q(\vec{d}) \in \mathcal{M} \mid \satisfy{\vec{d}}{\fnFilterExpr{q}} \} \cup \adatabase$.
	We show that $\mathcal{M'}$ is the model of $\aprogram'$ and $\adatabase$ as for Theorem~\ref{thmFilterPropagationCorrectness}:
	\begin{itemize}
		\item By definition, $\adatabase \subseteq \mathcal{M'}$.
		\item Showing that $\mathcal{M'}$ is closed under $\aprogram'$ is unaffected by allowing only conjunction in filter formulas.
		\item To show minimality of $\mathcal{M'}$, we use the same arguments as in the proof of Theorem~\ref{thmFilterPropagationCorrectness}, together with \eqref{eq_conjunction_generalises}.
		\qedhere
	\end{itemize}
\end{proof}

\thmCasfTractable*
\begin{proof}
	Since the arity of predicates in $\sigPredFilter$ is bounded, the set of filter atoms
	$\sigPredFilter[\arity{b}]$ is polynomial over the positional markers $\sigNumPos_k$ for any arity
	$k$ of a non-filter predicate.
	
	Claim: Whenever \eqref{eq_sf_fltconj} is applied to update $\fnFilterExpr{b}$ to $\fnFilterExpr{b}'$,
	we have $\fnFilterExpr{b}\models\fnFilterExpr{b}'$.
	Let $C_b,C_b'\subseteq\mathcal{A}$ denote the respective sets of conjuncts.
	For a Horn approximation $\mathcal{T}$, the condition in \eqref{eq_sf_fltconj} 
	is equivalent to (a) $\iota_{b(\vec{y})}(\fnFilterExpr{b})\approxmodels\iota_{b(\vec{y})}(A)$ and (b) $G\approxmodels\iota_{b(\vec{y})}(A)$
	holding individually, where (a) can be simplified to $\fnFilterExpr{b}\approxmodels A$.
	Since $\mathcal{T}$ is a Horn approximation of $\approxmodels$, 
	$F\approxmodels B$ for a conjunction $F$ is equivalent to $B$ being a logical consequence $\mathcal{T}$ and $F$.
	Therefore, $C_b'$ as computed in \eqref{eq_sf_fltconj} is closed under logical consequences
	from $\mathcal{T}$, i.e., if $A$ is a consequence of $C_b'$ and $\mathcal{T}$, then $A\in C_b'$.
	In particular, (a) further simplifies to $A\in C_b$, which shows $C_b'\subseteq C_b$ and
	therefore shows the claim.
	
	With filter formulas $\fnFilterExpr{b}$ corresponding to strictly decreasing sets of conjuncts from a
	polynomial set of filter atoms, the algorithm must terminate in polynomially many iterations.
	The computation in \eqref{eq_sf_fltconj} reduces to polynomially many checks of $\approxmodels$ relationships,
	each of which can be split into parts (a) and (b). Since (a) was found to be equivalent to $A\in C_b$,
	it can clearly be checked in polynomial time.
	For the remaining check (b) $G\approxmodels\iota_{b(\vec{y})}(A)$, we distinguish the cases in the theorem.
	
	Case 1: If $\mathcal{T}$ is a linear approximation, then we recursively compute a set $\mathcal{S}$ of ``necessarily false'' filter atoms
	by initialising $\mathcal{S}=\{\iota_{b(\vec{y})}(A)\}$ and applying rules of $\mathcal{T}$ backwards: 
	if rule $H\leftarrow B\in\mathcal{T}$ can be instantiated with a substitution $\sigma$ such that $H\sigma\in\mathcal{S}$, then the atom $B\sigma$ is
	added to $\mathcal{S}$. When this computation terminates, we create an expression $G'$ by replacing every atom $B$ in $G$ with $\bot$ if $B\in\mathcal{S}$,
	and with $\top$ if $B\notin\mathcal{S}$. The expression $G'$ only uses $\top$, $\bot$, $\wedge$, and $\vee$, and can be simplified to either $\top$ or $\bot$
	in polynomial time. If the result is $\top$, then $G\approxmodels\iota_{b(\vec{y})}(A)$ is true; otherwise it is false.
	
	Case 2: If the filter expressions $G_\sigPredFilter$ in $P$ do not contain $\vee$, then check (b) $G\approxmodels\iota_{b(\vec{y})}(A)$
	corresponds to a Datalog entailment check, since $G$ is a conjunction of filter atoms that one merely has to evaluate $\mathcal{T}$ over to
	decide (b). Since $\mathcal{T}$ is fixed, the complexity of this check is polynomial (the data complexity of Datalog).
\end{proof}

\propHornApproxCoNP*
\begin{proof}
	We reduce from the $\coNP$-complete problem of propositional logic unsatisfiability.
	Consider a propositional logic formula $\psi$ over propositional variables $p_1,\ldots,p_n$
	in negation normal form.
	For each $i\in\{1,\ldots,n\}$, consider unary filter predicates $t_i,f_i$ and
	set $r(p_i)=t_i(\numpos{1})$ and $r(\neg p_i)=f_i(\numpos{1})$.
	A filter formula $r(\psi)$ is obtained from $\psi$ by replacing each negated variable $\neg p_i$ by
	$r(\neg p_i)$ and each non-negated variable $p_i$ by $r(p_i)$.
	Then let $G=\bigwedge_i(r(p_i)\vee r(\neg p_i))\wedge r(\psi)$
	and let $\mathcal{T}=\{b(x)\leftarrow t_i(x)\wedge f_i(x)\mid 1\leq i\leq n\}$ for a filter predicate $b$
	not used elsewhere.
	
	Then $A$ is a consequence of $G$ and $\mathcal{T}$ iff $\psi$ is unsatisfiable.
	Indeed, if $\psi$ is satisfied by assignment $\beta$, we define
	$\beta'(t_i(\numpos{1}))=\mathit{true}$ iff $\beta(p_i)=\mathit{true}$,
	$\beta'(f_i(\numpos{1}))=\mathit{true}$ iff $\beta(p_i)=\mathit{false}$,
	and $\beta'(A)=\mathit{false}$.
	Then $\beta'$ satisfies $G$ and $\mathcal{T}$ by construction,
	so $A$ is indeed not a consequence of $G$ and $\mathcal{T}$.
	
	Conversely, if there is an assignment $\beta'$ that satisfies $G$ and
	$\mathcal{T}$ but not $A$, then we can find a satisfying assignment
	$\beta$ for $\psi$ by setting $\beta(p_i)=\mathit{true}$ iff $\beta'(t_i(\numpos{1}))=\mathit{true}$.
	Since $\beta'$ does not satisfy $A$, the premises of all rules
	of $\mathcal{T}$ are not satisfied either, hence
	$\beta'(t_i(\numpos{1}))=\mathit{true}$ and $\beta'(f_i(\numpos{1}))=\mathit{false}$ are not
	both true for any $i$. The fact that one of the two must be true follows since $G$ is satisfied.
\end{proof}

\section{Proofs for Section~\ref{sec_negation}}

We first state an auxiliary result for plain Datalog.

\begin{lemma}\label{lemmaAdmissibilityStrongEnough}
	Let $\aprogram$ be a Datalog program, let $\aprogram'$ be an admissible rewriting of $\aprogram$, and let $\sigPredIDB$ be the IDB predicates of $\aprogram$.
	For $p(\vec{c})$ with $p \in \sigPredIDB$, if $\aprogram',\adatabase \models p(\vec{c})$, then $\vec{c} \in \fnFilterExpr{p}^\adatabase$.
	
	This result also holds for any choice of filter formulas $\fnFilterExpr{p}$, even if not computed by applying Algorithm~\ref{alg_filter_pushing} to $\aprogram$,
	as long as the conditions of admissibility hold with these filters.
\end{lemma}

\begin{proof}
	Let $\inter{M}$ be the model of $P'$ and $\adatabase$.
	Let $\inter{M'} \coloneqq \{ p(\vec{t}) \in \inter{M} \mid \vec{t} \in \fnFilterExpr{p}^\adatabase \} \cup \adatabase$.
	Since $\adatabase \subseteq \inter{M}$, we have $\inter{M'} \subseteq \inter{M}$.
	
	Let $\arule' = h(\vec{x}) \leftarrow \shBodyNormal \land \pi_\psi \in \aprogram'$ be an admissible rewriting of $\arule \in \aprogram$ and
	let $\sigma$ be a mapping such that $\shBodyNormal\sigma \subseteq \inter{M'}$ and $\psi\sigma \subseteq \adatabase$.
	For each $b(\vec{y}) \in \shBodyNormal$ with $b \in \sigPredIDB$, $b(\sigma(\vec{y})) \notin \adatabase$ and, hence, $\sigma(\vec{y}) \in \fnFilterExpr{b}^\adatabase$.
	Let $F_-$ and $F_+$ be as in Def.~\ref{def_admissible}.
	We have $F_-\sigma \subseteq \adatabase$.
	By admissibility,  $F_+\sigma \subseteq \adatabase$.
	In particular, $\sigma(\vec{x}) \in \fnFilterExpr{h}^\adatabase$.
	Since $\inter{M}$ is model of $P'$ and $\adatabase$, we have $h(\sigma(\vec{x})) \in \inter{M}$ and, therefore, $h(\sigma(\vec{x})) \in \inter{M'}$.
	This shows that $\inter{M'}$ satisfies $\aprogram$ and $\adatabase$. By definition, $\inter{M'} \subseteq \inter{M}$, so 
	$\inter{M} = \inter{M'}$ since $\inter{M}$ is the least model.
\end{proof}

The following statement is again about normal logic programs.
Here we assume again that filters $\fnFilterExpr{p}$ were computed for $\aprogram$ as in Section~\ref{sec_negation}.

\begin{lemma}\label{lemmaAspRelevantNafAtomPresent}
	Let $\aprogram'$ be an admissible rewriting of $\aprogram$, 
	and let $\mathcal{M}$ be the least model of $\ground{\aprogram}^{\mathcal{C}}$ and $\mathcal{C}$
	for fact sets $\mathcal{B}$ and $\mathcal{C}$ such that
	(i) $\mathcal{B} \subseteq \mathcal{C}$,
	(ii) $\mathcal{B}$ is closed under $\ground{\aprogram'}^{\mathcal{B}}$, and
	(iii) for all $p(\vec{c})\in\mathcal{C}\setminus\mathcal{B}$, we have $\vec{c} \notin \fnFilterExpr{h}^\adatabase$.
	
	Then, for all $h(\vec{c})\in\mathcal{M}\setminus\mathcal{B}$, we have $\vec{c} \notin \fnFilterExpr{h}^\adatabase$.
\end{lemma}

\begin{proof}
	$\mathcal{M}$ can be obtained from $\mathcal{C}$ by a sequence of applications of rules $\arule_1\sigma_1, \arule_2\sigma_2, \ldots$ with $\arule_i\sigma_i \in \ground{\aprogram}^{\mathcal{C}}$, where we denote $\arule_i$ as $h_i(\vec{x_i}) \leftarrow  \shBodyNormali{i} \land \shBodyFilterGi{i}$.
	We show the lemma inductively, i.e., we show that if, for all $j < i$, $\sigma_j(\vec{x_j}) \notin \fnFilterExpr{h_j}^\adatabase$ if $h_j(\sigma_j(\vec{x_j})) \notin \mathcal{B}$, then $\sigma_i(\vec{x_i}) \notin \fnFilterExpr{h_i}^\adatabase$ if $h_i(\sigma_i(\vec{x_i})) \notin \mathcal{B}$.
	
	Assume for a contradiction that $\sigma_i(\vec{x_i}) \in \fnFilterExpr{h_i}^\adatabase$ with $h_i(\sigma_i(\vec{x_i})) \notin \mathcal{B}$.
Since $\arule_i\sigma_i$ is applicable,
	$\shBodyFilterGi{i}\sigma_i \subseteq \adatabase$.
	For $b(\vec{y}) \in \shBodyNormali{i}$, let $G$ and $M$ be defined as in Algorithm~\ref{alg_filter_pushing}.
	Since $G \models \iota_{b(\vec{y})}(M)$ and $\iota_{b(\vec{y})}(M) \models \iota_{b(\vec{y})}(\fnFilterExpr{b})$, 
	we have $\iota_{b(\vec{y})}(\fnFilterExpr{b})\sigma_i \subseteq \adatabase$, i.e., $\satisfy{\sigma_i(\vec{y})}{\fnFilterExpr{b}}$.
	By induction, $b(\sigma_i(\vec{y})) \in \mathcal{B}$ for all $j < i$.
	Let $\tau_i = \shHeadi{i} \leftarrow \shBodyNormali{i} \land \shBodyNafi{i} \land \shBodyFilterGi{i}$ be the rule $\tau_i\in\aprogram$ from which $\arule_i\in\ground{\aprogram}^{\mathcal{C}}$ stems, let $\tau'_i = \shHeadi{i} \leftarrow \shBodyNormali{i} \land \shBodyNafi{i} \land \psi_i$ be an admissible rewriting of $\tau_i$ such that $\tau'_i\in \aprogram'$, and let $\arule'_i = \shBodyNormali{i} \land \psi_i$.
	Since $\arule_i\sigma_i \in \ground{P}^{\mathcal{C}}$, we have $\shBodyNafi{i}\sigma_i \cap \mathcal{C} = \emptyset$ and, by $\mathcal{B} \subseteq \mathcal{C}$, $\shBodyNafi{i}\sigma_i \cap \mathcal{B} = \emptyset$,
	i.e., $\arule'_i\sigma_i \in \ground{P'}^{\mathcal{B}}$.
	Moreover, $\shBodyNormali{i}\sigma_i \subseteq \mathcal{B}$.
	Since $\sigma_i(\vec{x}_i) \in \fnFilterExpr{h_i}^\adatabase$, $\shBodyFilterG\sigma_i \subseteq \adatabase$, and $\tau'_i$ is an admissible rewriting of $\tau_i$, we have $F_+ \models \psi_i$ and $\psi_i\sigma_i \subseteq \adatabase$.
	Therefore, $\arule'_i\sigma_i$ is applicable for $\mathcal{B}$, which yields the required contradiction, since $\mathcal{B}$ is already closed under $\arule'_i\sigma_i$.
\end{proof}

\begin{lemma}\label{lemma_bijection_well_defined}
	Let $P'$ be an admissible rewriting of $\aprogram$ and let $\mathcal{A}$ be a stable model of $\aprogram$ and $\adatabase$.
	Then, $\mathcal{A'} = \{ p(\vec{c}) \in \mathcal{A} \mid \vec{c} \in \fnFilterExpr{p}^\adatabase \}$ is a stable model of $\aprogram'$ and $\adatabase$.
\end{lemma}

\begin{proof}
To show that $\mathcal{A'}$ is a stable model for $P'$ and $\adatabase$, we show that (i) $\mathcal{A'} \models \ground{P'}^{\mathcal{A'}}$ and (ii) $\mathcal{A'}_-$ is not the model of $\ground{P'}^{\mathcal{A'}}$ and $\adatabase$ for all $\mathcal{A'}_- \subset \mathcal{A'}$.
	
	$\mathcal{A'}$ is closed under $\ground{P'}^{\mathcal{A'}}$:
	Let $\arule' = p(\vec{x}) \leftarrow B_{\bar{\sigPredFilter}} \land B_{\bar{\sigPredFilter}}^- \land \psi$ be a rule in $\aprogram'$ with atoms $\shBodyNormal$, negated atoms $\shBodyNaf$, and filter formula $\psi$.
	Let $\tau' = \arule'\sigma$ be its grounding for a mapping $\sigma$
	such that $B_{\bar{\sigPredFilter}}\sigma \subseteq \mathcal{A'}$, $B_{\bar{\sigPredFilter}}^-\sigma \cap \mathcal{A'} = \emptyset$, and $\psi\sigma \subseteq \adatabase$.
	There is $\arule = p(\vec{x}) \leftarrow \shBodyNormal \land \shBodyNaf \land \shBodyFilterG \in \aprogram$ such that $\arule'$ is an admissible rewriting of $\arule$.
	Trivially, $B_{\bar{\sigPredFilter}}\sigma \subseteq \mathcal{A'} \subseteq \mathcal{A}$.
	By definition of $\mathcal{A'}$, for $b(\vec{d}) \in B_{\bar{\sigPredFilter}} \subseteq \mathcal{A'}$, we have $\vec{d} \in \fnFilterExpr{b}^\adatabase$.
	By admissibility of $\arule'$, $(\psi \land F_-)\sigma \subseteq F_+\sigma$ with $F_-$ and $F_+$ as in Def.~\ref{def_admissible}, and $F_+\sigma \subseteq \adatabase$.
	In particular, $G_\sigPredFilter\sigma \subseteq \adatabase$ and $\sigma(\vec{x}) \in \fnFilterExpr{p}^\adatabase$.
	Since the modified variant of Algorithm~\ref{alg_filter_pushing} for normal programs loops in line~\alglineref{line_fp_bodyloop} over all $\naf b(\vec{y}) \in B_{\bar{\sigPredFilter}}^-$,
	we have $\sigma(\vec{y}) \in \fnFilterExpr{b}^\adatabase$, i.e., $b(\sigma(\vec{y})) \notin \mathcal{A}$.
	Hence, $\shBodyNaf\sigma \cap \mathcal{A} = \emptyset$ and $\arule\sigma \in \ground{P}^\mathcal{A}$.
	Therefore, $p(\sigma(\vec{x})) \in \mathcal{A}$, as $\mathcal{A}$ is a stable model for $P$, and $p(\sigma(\vec{x}))\in \mathcal{A}'$, i.e., $\mathcal{A'} \models \ground{P'}^{\mathcal{A'}}$.
	
	Minimality of $\mathcal{A'}$:
	For a contradiction, suppose there is a non-empty set $\mathcal{N}\subseteq\mathcal{A'}$ such that
	$\mathcal{A'}_- = \mathcal{A'} \setminus \mathcal{N}$ is also a model of $\ground{P'}^{\mathcal{A'}}$ and $\adatabase$.
	In particular, $\adatabase\subseteq\mathcal{A'}_-$, so $\mathcal{N} \cap \adatabase = \emptyset$.
	Let $\mathcal{A}_- = \mathcal{A} \setminus \mathcal{N}$.
Since $\mathcal{A}$ is a stable model for $P$,
	there is $\arule = p(\vec{x}) \leftarrow \shBody \in P$ and mapping $\sigma$ such that $p(\sigma(\vec{x})) = p(\vec{t}) \in \mathcal{N}$,
	$\shBodyNormal\sigma \subseteq \mathcal{A}_-$, $\shBodyNaf\sigma \cap \mathcal{A} = \emptyset$, and $\shBodyFilterG\sigma \subseteq \adatabase$.
	In other words, $p(\vec{t})$ is among the ``first'' conclusions that $\mathcal{A}_-$ is missing to be a model.
	
	Now let $\arule' = p(\vec{t}) \leftarrow \shBodyNormal \land \shBodyNaf \land \psi \in P'$ be an admissible rewriting of $\arule$.	
	For $b(\vec{y}) \in \shBodyNormal$, let $G$ and $M$ be defined as in Algorithm~\ref{alg_filter_pushing}.
	Since $G \models \iota_{b(\vec{y})}(M)$ and $\iota_{b(\vec{y})}(M) \models \iota_{b(\vec{y})}(\fnFilterExpr{b})$, 
	we have $\iota_{b(\vec{y})}(\fnFilterExpr{b})\sigma \subseteq \adatabase$, i.e., $\satisfy{\sigma(\vec{y})}{\fnFilterExpr{q}}$.
	Therefore, $b(\sigma(\vec{y})) \in \mathcal{A}_-'$, and $\shBodyNormal\sigma \subseteq \mathcal{A}_-'$.
	Let $F_+$ be defined as in Def.~\ref{def_admissible}.
	Since $F_+\sigma \subseteq \adatabase$ and $\arule'$ is an admissible rewriting, $\psi\sigma \subseteq \adatabase$.
	Since $\mathcal{A'} \subseteq \mathcal{A}$, $\shBodyNaf \cap \mathcal{A'} = \emptyset$.
	Hence, $p(\vec{t}) \leftarrow \shBodyNormal\sigma \land \psi\sigma \in \ground{P'}^{\mathcal{A'}}$ is applicable for $\mathcal{A}'_-$, which yields the required contradiction.
\end{proof}

For the converse of Lemma~\ref{lemma_bijection_well_defined}, we construct a stable model 
$\mathcal{A} \in \stablemods{P}{\adatabase}$ from a stable model $\mathcal{A'} \in \stablemods{P'}{\adatabase}$.
The construction in this case requires more careful processing, re-considering additional rules in the order defined by the partial stratifiction of $P$.
We do not require $P$ to be fully stratified, but we can always partition it into
a stratified program $P^1\cup\ldots\cup P^n$ (with $n$ strata) and a remainder program $P^\ast$ that forms a final stratum 
above $P^1\cup\ldots\cup P^n$.

\begin{lemma}\label{lemma_bijection_surjective}
	Let $\aprogram'$ be an admissible rewriting of $P$ with a stable model $\mathcal{A}'\in\stablemods{\aprogram'}{\adatabase}$.
	Let $\xi : \sigPredStrat \to \{ 1, \ldots, n \}$ be a mapping such that
	(i) $\xi(p) \leq \xi(q)$ if $p \to_+ q$ and
	(ii) $\xi(p) < \xi(q)$ if $p \to_- q$.
	Let $\{ P^1, \ldots, P^n, P^\ast \}$ be a partition of $P$ such that $P^i = \{ p(\vec{x}) \leftarrow B \in P \mid p \in \sigPredStrat, \xi(p) = i \} $ and $P^\ast = \{ p(\vec{x}) \leftarrow B \in P \mid p \notin \sigPredStrat\} $.
	Let $\mathcal{A}_0 = \mathcal{A}'$, 
	let $\mathcal{A}_i$ be the model of $\ground{P^i}^{\mathcal{A}_{i-1}}$ and $\mathcal{A}_{i-1}$ for $i\in\{1,\ldots,n\}$,
	and let $\mathcal{A}$ be the model of $\ground{P^\ast}^{\mathcal{A}_n}$ and $\mathcal{A}_n$.
	
	Then $\mathcal{A}\in\stablemods{\aprogram}{\adatabase}$ and $\mathcal{A'} = \{ p(\vec{c}) \in \mathcal{A} \mid \vec{c} \in \fnFilterExpr{p}^\adatabase \}$.
\end{lemma}

\begin{proof}
	We first show that $\mathcal{A'} = \{ p(\vec{c}) \in \mathcal{A} \mid \vec{c} \in \fnFilterExpr{p}^\adatabase \}$ ($\ddagger$).
	By the definitions, $\adatabase \subseteq \mathcal{A'} = \mathcal{A}_0 \subseteq \ldots \subseteq \mathcal{A}_n \subseteq \mathcal{A}$.
	Now by repeated application of Lemma~\ref{lemmaAspRelevantNafAtomPresent} with $\mathcal{B} = \mathcal{A'}$,
	$\mathcal{C}$ iterating over $\mathcal{A}_0, \ldots, \mathcal{A}_n$, and
	$P$ iterating over $P^1, \ldots, P^n, P^\ast$, we get that, for all $p(\vec{c}) \in \mathcal{A} \setminus \mathcal{A'}$,
	$\vec{c} \notin \fnFilterExpr{p}^\adatabase$. This establishes claim ($\ddagger$).
	
	It remains to show that $\mathcal{A} \in \stablemods{P}{\adatabase}$, by verifying the relevant properties.

	\paragraph*{$\mathcal{A}$ closed under $\ground{P}^{\mathcal{A}}$:}
	Let $\arule = h(\vec{c}) \leftarrow \shBodyNormal \land \shBodyFilterG \in \ground{P}^{\mathcal{A}}$,
	such that $\mathcal{A}\models\shBodyNormal \land \shBodyFilterG$.
	If $h \in \sigPredStrat$ with $k = \xi(h)$, then $\mathcal{A}_k\models\shBodyNormal \land \shBodyFilterG$ and $\arule \in \ground{P^k}^{\mathcal{A}_{k-1}} \supseteq \ground{P}^{\mathcal{A}}$; hence, $h(\vec{c}) \in \mathcal{A}_k \subseteq \mathcal{A}$.
	Otherwise, $h \notin \sigPredStrat$ and $\arule \in \ground{P^\ast}^{\mathcal{A}_n} \supseteq \ground{\aprogram}^{\mathcal{A}}$; hence, $h(\vec{c}) \in \mathcal{A}$.
	Therefore, $\mathcal{A} \models \ground{P}^{\mathcal{A}}$.
	
	\paragraph*{Minimality of $\mathcal{A}$:}
	For a contradiction, suppose there is a non-empty set $\mathcal{N}\subseteq\mathcal{A}$ such that
	$\mathcal{A}_- = \mathcal{A} \setminus \mathcal{N}$ is also a model of $\ground{P}^{\mathcal{A}}$ and $\adatabase$.
	In particular, $\adatabase\subseteq\mathcal{A}_-$, so $\mathcal{N} \cap \adatabase = \emptyset$.
	\begin{itemize}
		\item Case 1: there is $p(\vec{c}) \in \mathcal{N}\cap\mathcal{A}'$.
			Since $\mathcal{A'}$ is the least model of $\ground{P'}^{\mathcal{A'}}$ and $\adatabase$, and since 
			$\adatabase\subseteq\mathcal{A'}\setminus\mathcal{N}$, we get $\mathcal{A'}\setminus\mathcal{N}\not\models\ground{P'}^{\mathcal{A'}}$.
		 	Therefore, there is $\tau' = \shHead \leftarrow \shBodyNormal \land \shBodyNaf \land \shBodyFilterGRewrite$ in $P'$ and mapping $\sigma$ such that
		 	$h(\vec{c}) = h(\sigma(\vec{x})) \in \mathcal{N}$, $\shBodyNormal\sigma \subseteq \mathcal{A'}\setminus\mathcal{N}$,
		 	$\shBodyNaf\sigma \cap \mathcal{A'} = \emptyset$, and
		 	$\shBodyFilterGRewrite\sigma \subseteq\adatabase$.
		 	There is $\tau \in P$ with $\tau = \shHead \leftarrow \shBodyNormal \land \shBodyNaf \land \shBodyFilterG$ such that $\tau'$ is an admissible rewriting of $\tau$.
Let $F_-$ and $F_+$ be as in Definition~\ref{def_admissible}.
		 	By Lemma~\ref{lemmaAdmissibilityStrongEnough}, for $b(\vec{y}) \in \shBodyNormal$ with IDB predicate $b$, $\sigma(\vec{y}) \in \fnFilterExpr{b}^\adatabase$.
		 	Therefore, $F_-\sigma \subseteq \adatabase$ and, by admissibility of $\tau'$, we have $F_+\sigma \subseteq \adatabase$,
		 	and in particular $\shBodyFilterG\sigma \subseteq\adatabase$ and $\sigma(\vec{x})\in\fnFilterExpr{h}^\adatabase$.
Moreover, $\shBodyNormal\sigma \subseteq \mathcal{A'}\setminus\mathcal{N}\subseteq\mathcal{A}_-$.
		 	
		 	It remains to show that $\shBodyNaf\sigma \cap \mathcal{A} = \emptyset$. 
		 	Therefore, consider an arbitrary $b(\vec{d})\sigma \in \shBodyNaf\sigma$.
		 	\begin{enumerate}[({1}a)]
		 	\item If $b$ is an EDB predicate, then $b(\vec{d})\notin\mathcal{A}$, since $\mathcal{A}$ and $\mathcal{A}'$ have the same EDB facts and
		 	$b(\vec{d}) \notin \mathcal{A'}$.
		 	\item If $b \in \sigPredStrat$, then $\vec{d} \in \fnFilterExpr{b}^\adatabase$. Indeed, the modified line~\alglineref{line_fp_bodyloop}
		 	in Algorithm~\ref{alg_filter_pushing} considers $b(\vec{y})$.
		 	From our earlier observation that $F_+\sigma \subseteq \adatabase$, we get that $G\sigma \subseteq \adatabase$
		 	for $G$ as in \alglineref{line_fp_matchfilter}, so the susequent update of $\fnFilterExpr{b}$ ensures $\vec{d} \in \fnFilterExpr{b}^\adatabase$.
		 	
Therefore, since $\vec{d} \in \fnFilterExpr{b}^\adatabase$, claim ($\ddagger$) implies that $b(\vec{d}) \notin \mathcal{A'} \subseteq \mathcal{A}$.
		 	\item If $b \notin \sigPredStrat$ is an IDB predicate, then $\vec{d} \in \fnFilterExpr{b}^\adatabase$ by the modified initialisation \eqref{eq_sf_init_asp} for line \alglineref{line_fp_init} of Algorithm~\ref{alg_filter_pushing}.
		 	Using ($\ddagger$) and $\vec{d} \in \fnFilterExpr{b}^\adatabase$, we have $b(\vec{d}) \notin \mathcal{A'} \subseteq \mathcal{A}$.
		 	\end{enumerate}	
		 	Hence, $\tau\sigma\in\ground{P}^{\mathcal{A}}$ but $\mathcal{A_-}\not\models\tau\sigma$ -- contradiction.
		\item Case 2: $\mathcal{N}\cap\mathcal{A}'=\emptyset$ and there is $p(\vec{c}) \in \mathcal{N}$ with $p \in \sigPredStrat$ and $k = \xi(p)$. W.l.o.g., assume that $k$ is minimal,
			and let $\mathcal{A}_k^-=\mathcal{A}_k\setminus\mathcal{N}$.
			By minimality of $k$ and $\mathcal{N}\cap\mathcal{A}'=\emptyset$, we have $\mathcal{A}_{k-1}\subseteq\mathcal{A}_k^-$.
			Yet $\mathcal{A}_k^-$
			is not a model of $\ground{P^k}^{\mathcal{A}_{k-1}}$ and $\mathcal{A}_{k-1}$, since it is strictly smaller than the least model $\mathcal{A}_k$.
			Therefore, $\mathcal{A}_k^- \not\models\ground{P^k}^{\mathcal{A}_{k-1}}$, so			
			there is $\arule = h(\vec{x}) \leftarrow \shBodyNormal \land \shBodyNaf \land \shBodyFilterG \in P^k$ and mapping $\sigma$ such that
			$h(\sigma(\vec{x})) \in \mathcal{N}$, $\shBodyNormal\sigma \subseteq \mathcal{A}_k^-$, $\shBodyNaf\sigma \cap \mathcal{A}_{k-1} = \emptyset$, and
			$\shBodyFilterG\sigma \subseteq \adatabase$.
So $\shBodyNormal\sigma \subseteq \mathcal{A}_k^- \subseteq \mathcal{A_-}$.
			For $b(\vec{d}) \in \shBodyNaf\sigma$, there are two cases:
			\begin{enumerate}[({2}a)]
		 	\item Same as (1a).
		 	\item If $b$ is an IDB predicate, then $b\in\sigPredStrat$ with $\xi(b) \leq k-1$. Again, $b(\vec{d}) \notin \mathcal{A}$, since $\mathcal{A}$ and $\mathcal{A}_{k-1}$ have the same facts for predicates of lower strata, and $b(\vec{d}) \notin \mathcal{A}_{k-1}$.
		 	\end{enumerate}			
			Hence, $\shBodyNaf\sigma \cap \mathcal{A} = \shBodyNaf\sigma \cap \mathcal{A}_{k-1} = \emptyset$ and $\arule\sigma \in \ground{P}^{\mathcal{A}}$.
			Therefore, $\mathcal{A_-}$ is not closed under $\ground{P}^{\mathcal{A}}$ -- contradiction.
		\item Case 3: $\mathcal{N}\cap\mathcal{A}'=\emptyset$ and for all $p(\vec{c}) \in \mathcal{A}$ with $p \in \sigPredStrat$, $p(\vec{c}) \notin \mathcal{N}$.
			Then $\mathcal{A}_n\subseteq\mathcal{A}_-$.
			As before, since $\mathcal{A}$ is the least model of $\ground{P^\ast}^{\mathcal{A}_n}$ and $\mathcal{A}_n$,
			$\mathcal{A}_-\not\models\ground{P^\ast}^{\mathcal{A}_n}$, so
		 	there is a rule $\tau = \shHead \leftarrow \shBodyNormal \land \shBodyNaf \land \shBodyFilterG$ with $h \notin \sigPredStrat$ and $\tau\in\aprogram$, and mapping $\sigma$ such that $h(\sigma(\vec{x})) \in \mathcal{N}$, $\shBodyNormal\sigma \subseteq \mathcal{A}_-$, $\shBodyNaf\sigma \cap \mathcal{A}_{n} = \emptyset$, and
		 	$\shBodyFilterG\sigma \subseteq \adatabase$.
We still require that $\shBodyNaf\sigma \cap \mathcal{A} = \emptyset$.
		 Therefore, consider any $\naf b(\vec{d}) \in \shBodyNaf\sigma$. We show that $b(\vec{d}) \notin \mathcal{A}$:
		 \begin{enumerate}[({3}a)]
		 	\item Same as (1a).
		 	\item Same as (2b), using $n$ instead of $k-1$.
\item Same as (1c).
\end{enumerate}
		 Hence, $\arule\sigma \in \ground{P}^{\mathcal{A}}$, which yields the required contradiction, since $\arule\sigma$ is applicable for $\mathcal{A}_-$.\qedhere
	\end{itemize}
\end{proof}

\begin{lemma}\label{lemma_bijection_injective}
	Let $P'$ be an admissible rewriting of $P$,
	let $\mathcal{A'}$ be a stable model of $P'$ and $\adatabase$, and
	let $P^1,\ldots, P^n, P^\ast$ and $\mathcal{A}_0, \ldots, \mathcal{A}_n, \mathcal{A}$ as in Lemma~\ref{lemma_bijection_surjective}.
	If $\mathcal{B} \in \stablemods{P}{\adatabase}$ and $\{p(\vec{c}) \in \mathcal{B} \mid \vec{c} \in \fnFilterExpr{p}^\adatabase\} = \mathcal{A'}$, then $\mathcal{A} = \mathcal{B}$.
\end{lemma}

\begin{proof}
	Let $\mathcal{B}_i = \{ p(\vec{c}) \in \mathcal{B} \mid \xi(p) \leq i \} \cup \mathcal{A'}$ for $0 \leq i \leq n$.
	We show $\mathcal{A}_i = \mathcal{B}_i$ by induction on $i$:
	\begin{itemize}
		\item By definition, $\mathcal{B}_0 = \mathcal{A}' = \mathcal{A}$.
		\item Assume that $\mathcal{B}_i = \mathcal{A}_i$.
			$\mathcal{A}_{i+1}$ is the model of $\ground{P^{i+1}}^{\mathcal{A}_i}$ and $\mathcal{A}_i$.
			By induction, $\mathcal{A}_{i+1}$ is the model of $\ground{P^{i+1}}^{\mathcal{B}_i}$ and $\mathcal{B}_i$.
			Moreover, $\mathcal{B}$ is the model of $\ground{P}^{\mathcal{B}}$ and $\mathcal{B}$.
			We have $\mathcal{B}_{i+1}$ is the model of $\ground{P}^{\mathcal{B}_i}$ and $\mathcal{B}_i$, since each $P^k$ defines the predicates $q$ with $\xi(q) = k$.
			Hence, $\mathcal{A}_{i+1} = \mathcal{B}_{i+1}$.
	\end{itemize}
	
	$\mathcal{A}$ is the model of $\ground{P^\ast}^{\mathcal{A}_n}$ and $\mathcal{A}_n$.
	$\mathcal{B}$ is the model of $\ground{P}^{\mathcal{B}}$ and $\mathcal{B}$.
	Since $\mathcal{B}$ is stable model of $P$ and $\adatabase$, and $\adatabase \subseteq \mathcal{B}_n \subseteq \mathcal{B}$,
	we have $\mathcal{B}$ is the model of $\ground{P}^{\mathcal{B}}$ and $\mathcal{B}_n$.
	$P^\ast$ defines exactly the predicates $p \notin \sigPredStrat$,
	so $\mathcal{B}$ is the model of $\ground{P^\ast}^{\mathcal{B}_n}$ and $\mathcal{B}_n$.
	By $\mathcal{A}_n = \mathcal{B}_n$, we have $\mathcal{A} = \mathcal{B}$.
\end{proof}

\theoAspRewritingCorrect*
\begin{proof}
	Let $\mu$ be the mapping $\stablemods{\aprogram}{\adatabase} \to \stablemods{\aprogram'}{\adatabase} \colon \mathcal{A}\mapsto\{p(\vec{c}) \in \mathcal{A} \mid \vec{c} \in \fnFilterExpr{p}^\adatabase\}$.
	By Lemma~\ref{lemma_bijection_well_defined}, $\mu$ is well-defined.
	By Lemma~\ref{lemma_bijection_surjective}, $\mu$ is surjective.
	By Lemma~\ref{lemma_bijection_injective}, $\mu$ is injective.
	Hence, $\mu$ is a bijection.
\end{proof}
 \end{document}